\documentclass{lmcs}
\pdfoutput=1

\usepackage{lastpage}
\lmcsdoi{21}{3}{12}
\lmcsheading{}{\pageref{LastPage}}{}{}%
{Nov.~14,~2024}{Jul.~29,~2025}{}


\usepackage{xspace,xcolor,dsfont,amssymb,amsthm}
\usepackage[hidelinks]{hyperref}
\usepackage[utf8]{inputenc}

\begin{document}
\title{Synthesis with Privacy Against an Observer}
\thanks{This research is supported by the Israel Science Foundation, Grant 2357/19, and the European Research Council, Advanced Grant ADVANSYNT}
\titlecomment{\lsuper*A preliminary version of this paper appears in the Proceedings of the 27th Foundations of Software Science and Computation Structures, FoSSaCS 2024.} 
\author[O.~Kupferman]{Orna Kupferman\lmcsorcid{0000-0003-4699-6117}}
\author[O.~Leshkowitz]{Ofer Leshkowitz\lmcsorcid{0000-0001-9225-2325}}
\author[N.~Shamash Halevy]{Namma Shamash Halevy\lmcsorcid{0009-0009-4736-6708}}
	
\address{School of Computer Science and Engineering, The Hebrew University of Jerusalem, Israel}
	
\newcommand{\stam}[1]{}
\newcommand{\sem}[1]{[\![#1]\!]} 
\newcommand{\zug}[1]{\langle #1 \rangle}
\newcommand{\set}[1]{{\left\{ #1 \right\}}}
\newcommand{\E}[1]{\mathbb{E}[ #1 ]}
\renewcommand{\P}{\mathcal{P}}
\newcommand{\floor}[1]{\left\lfloor #1 \right\rfloor}
\newcommand{\ceil}[1]{\left\lceil #1 \right\rceil}

\newcommand{\A}{\mathcal{A}}
\newcommand{\N}{\mathcal{N}}
\newcommand{\B}{\mathcal{B}}
\newcommand{\C}{\mathcal{C}}
\newcommand{\D}{\mathcal{D}}
\newcommand{\T}{\mathcal{T}}
\newcommand{\G}{\mathcal{G}}
\newcommand{\U}{\mathcal{U}}
\renewcommand{\H}{\mathcal{H}}
\newcommand{\V}{\mathcal{V}}
\newcommand{\spec}{\varphi}
\renewcommand{\sec}{\psi}
\newcommand{\trig}{\theta}
\newcommand{\req}{\mathsf{req}}
\newcommand{\grant}{\mathsf{grant}}
\newcommand{\wait}{\mathsf{wait}}
\newcommand{\change}{\mathsf{change}}

\renewcommand{\fill}{\mathsf{fill}}
\newcommand{\fillH}{\fill_\H}
\newcommand{\hide}{\mathsf{hide}}
\newcommand{\hideH}{\hide_\H}
\newcommand{\noise}[1]{\mathsf{noise}_{#1}}
\newcommand{\noiseH}{\noise{\H}}
\newcommand{\LTL}{LTL\xspace}
\newcommand{\cost}{\mathsf{cost}}
\renewcommand{\H}{{\mathcal H}}
\newcommand{\Nat}{{\mathbb N}}
\newcommand{\True}{\mathtt T}
\newcommand{\False}{\mathtt F}

\newcommand{\Next}{\mathsf{X}}
\newcommand{\Ev}{\mathsf{F}}
\newcommand{\Alw}{\mathsf{G}}
\newcommand{\Until}{\mathsf{U}}
\newcommand{\sep}{~ | ~}

\begin{abstract}
	We study automatic {\em synthesis\/} of systems that interact with their environment and maintain {\em privacy\/} against an observer to the interaction.  The system and the environment interact via sets $I$ and $O$ of input and output signals. The input to the synthesis problem contains, in addition to a specification, also a list of {\em secrets\/}, a function $\cost: I \cup O \rightarrow \Nat $, which maps each signal to the cost of hiding it, and a bound $b \in \Nat$ on the budget that the system may use for hiding of signals. The desired output is an $(I/O)$-transducer $\T$ and a set $\H \subseteq I \cup O$ of signals that respects the bound on the budget, thus $\sum_{s \in \H} \cost(s) \leq b$, such that for every possible interaction of $\T$, the generated computation satisfies the specification, yet an observer from which the signals in $\H$ are hidden, cannot evaluate the secrets.  
	
	We first show that the complexity of the problem is 2EXPTIME-complete for specifications and secrets in LTL, thus it is not harder than synthesis with no privacy requirements. We then analyze the complexity of the problem more carefully, isolating the two aspects that do not exist in traditional synthesis, namely the need to hide the value of the secrets and the need to choose the set $\H$. We do this by studying settings in which traditional synthesis can be solved in polynomial time -- when the specification formalism is deterministic automata and when the system is closed, and show that each of the two aspects involves an exponential blow-up in the complexity. We continue and study {\em bounded synthesis with privacy}, where the input also includes a bound on the size of the synthesized transducer, as well as a variant of the problem in which the observer has 
	{\em knowledge}, either about the specification or about the system, 
	which can be helpful in evaluating the secrets.
	In addition, we study {\em certified privacy}, where the synthesis algorithm also provides to the user a certification about the secrets being hidden. 

\end{abstract}

\maketitle
\section{Introduction}
\label{intro}

{\em Synthesis\/} is the automated construction of correct systems from their specifications \cite{BCJ18}.
While synthesized systems are correct, there is no guarantee about their {\em quality}.
Since designers will be willing to give up manual design only after being convinced that the automatic process replacing it generates systems of comparable quality, it is extremely important to develop and study quality measures for automatically-synthesized systems. 
An important quality measure is {\em privacy}: making sure that the system and its environment do not reveal information they prefer to keep private. Privacy is a vivid research area in Theoretical Computer Science.
There, the notion of {\em differential privacy\/} is used for formalizing when an algorithm maintains privacy. Essentially, an algorithm is differentially private if by observing its output, one cannot tell if a particular individual's information is used in the computation \cite{DN03,DMNS16}. Another related notion is {\em obfuscation} in system development, where we aim to  develop systems whose internal operation is hidden \cite{BGIRSVY12,GGH0SW16}. Obfuscation has been mainly studied in the context of software, where it has exciting connections with cryptography \cite{BGIRSVY12,GGH0SW16}.

In the setting of automated synthesis in formal methods,
a very basic notion of privacy has been studied by means of synthesis with {\em incomplete information}
\cite{Rei84,KV00a,CDHR06}. There, the system should satisfy its specification eventhough it only has a partial view of the environment.
Lifting differential privacy to formal methods, researchers have introduced the temporal logic {\em HyperLTL}, which extends \LTL with explicit trace quantification \cite{CFKMRS14}. Such a quantification can relate computations that differ only in non-observable elements, and can be used for specifying that computations with the same observable input have the same observable output. The synthesis problem of HyperLTL is undecidable, yet is decidable for the fragment with a single existential quantifier, which can specify interesting properties  \cite{FHLST20}.
In \cite{KL22}, the authors suggested a general framework for automated synthesis of privacy-preserving reactive systems. In their framework, the input to the synthesis problem includes, in addition to the specification, also {\em secrets}.  During its interaction with the environment, the system may keep private some of the assignments to the output signals, and it directs the environment which assignments to the input signals it should keep private. Consequently, the satisfaction value of the specification and secrets may become unknown. The goal is to synthesize a system that satisfies the specification yet keeps the value of the secrets unknown. Finally, lifting obfuscation to formal methods, researchers have studied the synthesis of obfuscation policies for temporal specifications. In \cite{WRRLS18}, an obfuscation mechanism is based on edit functions that alter the output of the system, aiming to make it impossible for an observer to distinguish between secret and non-secret behaviors. In \cite{JPH10}, the goal is to synthesize a control function that directs the user which actions to disable, so that the  observed sequence of actions would not disclose a secret behavior.

In this paper we continue to study privacy-preserving reactive synthesis. As in \cite{KL22}, our setting is based on augmenting the specification with secrets whose satisfaction value should remain unknown. Unlike \cite{KL22},
the system and the environment have complete information about the assignments to the input and output signals, and the goal is to hide the secrets from a third party, and to do so by hiding the assignment to some of the signals throughout the interaction. As an example, consider a system that directs a robot patrolling a warehouse storage. Typical specifications for the system require it to direct the robot so that it eventually reaches the shelves of requested items, it never runs out of energy, etc. An observer to the interaction between the system and the robot may infer properties we may want to keep private, like dependencies between customers and shelves visited, locations of battery docking stations, etc. If we want to prevent the observer from
inferring these properties (a.k.a,. the secrets), we have to hide the interaction from it. Different effort should be made in order to hide different components of the interaction (alarm sound, content of shelves, etc.). Our framework synthesizes a system that realizes the specification without the secrets being revealed, subject to restrictions on hiding of signals.  
As another example, consider a scheduler that should grant access to a joint resource. The scheduler should maintain mutual exclusion (grants are not given to different users simultaneously) and non-starvation (all requests are granted), while hiding details like waiting time or priority to specific users. In Examples~\ref{scheduler} and~\ref{robot}, we describe in detail the application of our framework for the synthesis of such a scheduler, as well as its application in the synthesis of a robot that paints parts of manufactured pieces. The robot should satisfy some requirements about the generated pattern of colors while hiding other features of the pattern.

Formally, we consider a reactive system that interacts with its environments via sets $I$ and $O$ of input and output signals. At each moment in time, the system reads a truth assignment, generated by the environment, to the signals in $I$, and it generates a truth assignment to the signals in $O$. The interaction between the system and its environment generates a {\em computation}. The system {\em realizes\/} a specification $\spec$ if all its computations satisfy $\spec$ \cite{PR89a}.  
We introduce and study the problem of {\em synthesis with privacy in the presence of an observer}. Given a specification $\spec$, and secrets $\sec_1,\ldots,\sec_k$ over $I \cup O$, our goal is to return, in addition to a system that realizes the specification $\spec$, also a set $\H \subseteq I \cup O$ of {\em hidden signals}, such that the satisfaction value of the secrets $\sec_1,\ldots,\sec_k$ is unknown to an observer that does not know the truth values of the signals in $\H$.
Thus, secrets are evaluated according to a {\em three-valued semantics}. 
The use of secrets enables us to hide {\em behaviors}, rather than just signals. \footnote{Hiding of signals is a special case of our framework. Specifically, hiding of a signal $p$ can be done with the secrets $Fp$ and $F\neg p$.}
Obviously, hiding all signals guarantees that the satisfaction value of every secret is unknown.
Hiding of signals, however, is not always possible or involves some cost. We formalize this by adding to the setting a function $\cost: I \cup O \rightarrow \Nat $, which maps each signal to the cost of hiding its value, and a bound $b \in \Nat$ on the budget that the system may use for hiding of signals. The set $\H$ of hidden signals has to respect the bound, thus $\sum_{s \in \H} \cost(s) \leq b$.

In some cases, it is desirable to hide the truth value of a secret only when some condition holds. For example, we may require to hide the content of selves only in some sections of the warehouse. We extend our framework to {\em conditional secrets}: pairs of the form $\zug{\trig,\sec}$, where the satisfaction value of the secret $\sec$ should be hidden from the observer only when the trigger $\trig$ holds.  In particular, when $\trig=\sec$, we require to hide the secret only when it holds. For example, we may require to hide an unfair scheduling policy only when it is applied. Note that a conditional secret $\zug{\trig,\sec}$ is not equivalent to a secret $\trig \rightarrow \sec$ or $\trig \rightarrow \neg \sec$, and that the synthesized system may violate the trigger, circumventing the need to hide the secret. For example, by synthesizing a fair scheduler, the designer circumvents the need to hide an unfair policy.

We show that synthesis with privacy is 2EXPTIME-complete for specifications and secrets in \LTL. Essentially, once the set $\H$ of hidden signals is determined, we can compose an automaton for the specification with automata that verify, for each secret, that the assignments to the signals in $(I \cup O) \setminus \H$ can be completed both in a way that satisfies the secret and in a way that does not satisfy it. A similar algorithm works for conditional secrets.

While the complexity of our algorithm is not higher than that of \LTL synthesis with no privacy, it would be misleading to conclude that handling of privacy involves no increase in the complexity. The 2EXPTIME complexity follows from the need to translate \LTL specifications to deterministic automata on infinite words. Such a translation involves a doubly-exponential blow-up \cite{KV05b,KR10}, which possibly dominates other computational tasks of the algorithm. In particular, two aspects of synthesis with privacy that do not exist in usual synthesis are a need to go over all possible choices of signals to hide, and a need to go over all assignments to the hidden signals. 

Our main technical contribution is a finer complexity analysis of the problem, which reveals that each of the two aspects above involves an exponential complexity: the first in the number of signals and the second in the size of the secret. We start with the need to go over all assignments of hidden signals and show that even when the specification is $\True$, the set $\H$ of hidden signals is given, and there is only one secret, given by a deterministic B\"uchi automaton, synthesis with privacy is EXPTIME-complete. This is exponentially higher than synthesis of deterministic B\"uchi automata, which can be solved in polynomial time. 
We continue to the need to go over all possible choices of $\H$. For that, we focus on the closed setting, namely when $I=\varnothing$, and the case the specification and secrets are given by deterministic automata. We show that while synthesis with privacy can be then solved in polynomial time for  a given set $\H$, it is NP-complete when $\H$ is not given, even when the function $\cost$ is uniform.

We continue and study three variants of the problem: {\em bounded synthesis\/}, {\em certified privacy}, and {\em knowledgeable observer}, briefly described below. 
One way to cope with the 2EXPTIME complexity of LTL synthesis, which is carried over to a doubly-exponential lower bound on the size of the generated system \cite{Ros92}, is bounded synthesis. There, the input to the problem includes also a bound on the size of the system \cite{SF07,Ehl10,KLVY11}. In a setting with no privacy, the bound reduces the complexity of LTL synthesis to PSPACE, as one can go over all candidate systems. We study bounded synthesis with privacy and show that privacy makes the problem much harder: it is EXPSPACE-complete when the specification and secrets are given by LTL formulas, and is PSPACE-complete when they are given by deterministic parity (or B\"uchi) automata. 

The second variant we consider is one in which in addition to hiding the secret, the system has to generate a {\em certificate for privacy}. The certificate is an assignment to the hidden signals that differs from the assignment generated during the interaction and with which the secret has a different satisfaction value. This way, the user gets a witness computation that proves privacy is preserved. We examine the difference from the standard setting and specifically show that the need to generate certificates in an online manner makes specification and secrets harder to realize. Then, we provide a solution for synthesis with certified privacy and show that it is 2EXPTIME-complete and can be solved Safralessly (i.e., without determinization of a nondeterministic B\"uchi automaton) \cite{KV05c}.

Finally, recall that a system keeps a secret $\sec$ private if an observer cannot reveal the truth value of $\sec$: every observable computation can be completed both to a computation that satisfies $\sec$ and to a computation that does not satisfy $\sec$. We study a setting in which the observer has additional knowledge. First, we consider the setting in which the observer knows the specification $\spec$ of the system. Consequently, the observer knows that only completions that satisfy $\spec$ should be taken into account. If, for example, $\spec\rightarrow\sec$, then $\sec$ cannot be kept private. We describe an algorithm for this variant of the problem and analyze the way knowledge of the specification affects the complexity. In particular, we show that the problem becomes EXPTIME-complete even when the specification is given by a deterministic B\"uchi automaton and the secrets are of a fixed size. Then, we consider the even more demanding setting in which the observer knows the structure of the system. Consequently, the observer knows the set of possible computations of the system and thus may restrict to them when trying to evaluate the secret. We describe a procedure for checking whether a given $(I/O)$-transducer $\T$ hides a given secret from an observer that knows $\T$. We leave the synthesis problem for this variant open and briefly explain why we suspect it is undecidable.  

\section{Preliminaries}
\label{sec:preliminaries}

\subsection{Synthesis}
For a finite nonempty alphabet $\Sigma$, an infinite {\em word\/} $w = \sigma_0 \cdot \sigma_1\cdot\ldots \in \Sigma^\omega$ is an infinite sequence of letters from $\Sigma$. 
A {\em language\/} $L\subseteq \Sigma^\omega$ is a set of infinite words. 

Let $I$ and $O$ be disjoint finite sets of input and output signals, respectively. We consider the alphabet $2^{I\cup O}$ of truth assignments to the signals in $I\cup O$. 
Then, a languages $L \subseteq (2^{I \cup O})^\omega$ can be viewed as a {\em specification}, and the {\em truth value\/} of $L$ in a computation $w \in (2^{I\cup O})^{\omega}$ is $\True$ if $w \in L$, and is $\False$ otherwise. 

An {\em $(I/O)$-transducer\/} is a tuple $\T=\zug{I,O,S,s_0,\eta,\tau}$, where $S$ is a finite set of states, $s_0\in S$ is an initial state, $\eta:S\times 2^I\rightarrow S$ is a transition function, and $\tau: S\rightarrow 2^O$ is a labeling function.
We extend the transition function $\eta$ to words in $(2^I)^*$ in the expected way, thus $\eta^*: S\times (2^I)^{*}\rightarrow S$ is such that 
for all $s\in S$, $x_I \in (2^I)^{*}$, and $i\in 2^I$,
we have that $\eta^*(s,\epsilon)=s$,
and $\eta^*(s,x_I\cdot i)=\eta(\eta^*(s,x_I),i)$.
For a word $w_I=i_0\cdot i_1\cdot i_2\cdot\ldots\in (2^I)^{\omega}$, we define the \emph{computation of ~$\T$ on $w_I$} to be the word $\T(w_I)=(i_0\cup o_0)\cdot (i_1\cup o_1)\cdot\ldots \in (2^{I\cup O})^{\omega}$, where for all $j\geq 0$, we have that $o_j=\tau(\eta^*(s_0, i_0\cdots i_j))$.
The {\em language\/} of $\T$, denoted $L(\T)$, is the set of computations of  $\T$, that is $L(\T)=\{\T(w_I) : w_I\in (2^I)^{\omega}\}$.

We say that $\T$ \emph{realizes} a language $L\subseteq (2^{I\cup O})^{\omega}$ if $L(\T)\subseteq L$.
We say that a language $L\subseteq (2^{I\cup O})^\omega$ is \emph{realizable} if there is an $(I/O)$-transducer that realizes it.
In the {\em synthesis\/} problem, we are given a specification language $L\subseteq (2^{I\cup O})^\omega$ and we have to  return an $(I/O)$-transducer that realizes $L$ or decide that $L$ is not realizable. The language $L$ is given by an automaton over the alphabet $2^{I\cup O}$ or a temporal logic formula over $I \cup O$ (see definitions in Section~\ref{aut and ltl}).

\subsection{Synthesis with privacy}

In the {\em synthesis with privacy\/} problem, we are given, in addition to the specification language $L_\spec\subseteq (2^{I\cup O})^{\omega}$, also a {\em secret\/}  $L_\sec\subseteq (2^{I\cup O})^{\omega}$, which defines a behavior that we want to hide from an observer\footnote{See Remark~\ref{rem ms} for an extension of the setting to multiple and conditional secrets.}. Thus, we seek an $(I/O)$-transducer that realizes $L_\spec$ without revealing the truth value of $L_\sec$ in the generated computations. Keeping the truth value of $L_\sec$ secret is done by hiding the truth value of some signals in $I \cup O$. Before we define synthesis with privacy formally, we first need some notations.

Consider a set $\H\subseteq {I\cup O}$ of \emph{hidden signals}. Let $\V=({I\cup O})\setminus \H$ denote the set of \emph{visible signals}.
For an assignment $\sigma \in 2^{I\cup O}$, let $\sigma|_\V\in 2^\V$ and $\sigma|_\H$ be the restriction of $\sigma$ to the visible and hidden signals respectively. Thus, $\sigma=\sigma|_\V\cup \sigma|_\H$. 
Then, for an infinite word $w=\sigma_0\cdot \sigma_1\cdot \sigma_2\cdots \in (2^{I\cup O})^\omega$, we denote by $w|_\V$ and $w|_\H$, the restriction of $w$ to $\V$ and $\H$ respectively. Thus, $w|_\V=\sigma_0|_\V\cdot \sigma_1|_\V\cdot \sigma_2|_\V\cdots \in (2^\V)^\omega$ and $w|_\H=\sigma_0|_\V\cdot \sigma_1|_\V\cdot h_2\cdots \in (2^\H)^\omega$. 
Then, let $\noiseH(w)$ be the set of all computations that differ from $w$ in assignments to the signals in $\H$. 
Thus, when the signals in $\H$ are hidden, then an observer of a computation $w\in (2^{I\cup O})^\omega$ only knows that the computation is one from $\noiseH(w)$.  Formally, $\noiseH(w)=\set{w'\in (2^{I\cup O})^\omega:w'|_\V=w|_\V}$.

Consider a specification $L_\spec\subseteq (2^{I\cup O})^{\omega}$ and a secret $L_\sec\subseteq (2^{I\cup O})^{\omega}$.
For a set $\H\subseteq I\cup O$ of hidden signals, we say that an $(I/O)$-transducer $\T$ \emph{$\H$-hides\/} $L_\sec$ if for all words $w_I\in (2^I)^{\omega}$, the truth value of the secret $L_\sec$ in the computation $\T(w_I)$ cannot be deduced from $\T(w_I)|_\V$.
Formally, for every $w_I\in (2^I)^\omega$, there exist two computations $w^+,w^-\in \noiseH(\T(w_I))$, such that $w^+\in L_\sec$ and $w^-\notin L_\sec$.
We say that $\T$ \emph{realizes $\zug{L_\spec,L_\sec,\H}$ with privacy} if $\T$ realizes $L_\spec$ and $\H$-hides $L_\sec$. We say that $\zug{L_\spec,L_\sec,\H}$ is \emph{realizable with privacy\/} if there exists an $(I/O)$-transducer that realizes $\zug{L_\spec,L_\sec,\H}$ with privacy. 

Clearly, hiding is monotone with respect to $\H$, in the sense that the larger $\H$ is, the more likely it is for an $(I/O)$-transducer  $\T$ to $\H$-hide $L_\sec$.
Indeed, if $\T$ $\H$-hides $L_\sec$, then $\T$ $\H'$-hides $L_\sec$ for all $\H'$ with $\H\subseteq \H'$. In particular, 
taking $\H=I \cup O$, we can hide all non-trivial secrets. Hiding of signals, however, is not always possible, and may sometimes involve a cost. Formally, we consider a {\em hiding cost function\/} $\cost:I\cup O\rightarrow \Nat$, which maps each signal to the cost of hiding it, and a {\em hiding budget\/} $b\in \Nat$, which bounds the cost that the system may use for hiding of signals.  The cost of hiding a set $\H\subseteq I\cup O$ of signals is then $\cost(\H)=\sum_{p\in \H}\cost(p)$, and we say that $\H$ {\em respects\/} $b$ if  $\cost(\H)\leq b$. Note that if $\cost(p) > b$, for $p \in I \cup O$, then $p$ cannot be hidden. Also, when $\cost(p)=1$ for all $p\in I\cup O$, we say that $\cost$ is {\em uniform}. Note that then, $b$ bounds the number of signals we may hide. 

Now, we say that 
$\zug{L_\spec,L_\sec,\cost,b}$ is \emph{realizable with privacy} if there exists a set $\H\subseteq I\cup O$ such that $\H$ respects $b$ and $\zug{L_\spec,L_\sec,\H}$ is realizable with privacy.
Finally, in the {\em synthesis with privacy\/} problem, we are given $L_\spec$, $L_\sec$, $\cost$, and $b$, and we have to return a set $\H\subseteq I\cup O$ that $\H$ respects $b$ and an $(I/O)$-transducer $\T$ that realizes $\zug{L_\spec,L_\sec,\H}$ with privacy, or determine that $\zug{L_\spec,L_\sec,\cost,b}$ is not realizable with privacy.

\subsection{Multiple and conditional secrets}
\label{rem ms}
In this section we discuss two natural extensions of our setting.
First, often we need to hide form the observer more than one secret. We extend the definition of synthesis with privacy to a set of secrets $S=\set{L_{\sec_1}, L_{\sec_2}, \ldots, L_{\sec_k}}$ in the natural way. Thus, an $(I/O)$-transducer $\T$ realizes $\zug{\spec,S,\H}$ with privacy if it realizes $\spec$ and $\H$-hides $L_{\sec_i}$, for all $i\in [k]$. Note that 

Then, a {\em conditional secret\/} is a pair $\zug{L_\trig,L_\sec}$, consisting of a {\em trigger\/} and a {\em secret}. The truth value of the secret should be unknown only in computations that satisfy the trigger. Formally, for a set $\H\subseteq I\cup O$ of hidden signals, we say that an $(I/O)$-transducer $\T$ \emph{$\H$-hides\/} $\zug{L_\trig,L_\sec}$ if for all input sequences $w_I\in (2^I)^{\omega}$ such that $\noiseH(\T(w_I)) \subseteq L_\trig$, the truth value of $L_\sec$ in the computation $\T(w_I)$ cannot be deduced from $\T(w_I)|_\V$, thus there exist two computations $w^+,w^-\in \noiseH(\T(w_I))$, such that $w^+\in L_\sec$ and $w^-\notin L_\sec$. 
A useful special case of conditional secrets is when the trigger and the secret coincide, and so we have to hide the truth value of the secret only if there are computations where the value of secret is $\True$. Formally, $\T$ $\H$-hides $\zug{L_\sec,L_\sec}$ if for all input sequences $w_I\in (2^I)^{\omega}$, there exists a computation $w^-\in \noiseH(\T(w_I))$ such that $w^-\notin L_\sec$. 

Note that unlike a collection of specifications, which can be conjuncted, hiding a set of secrets is not equivalent to hiding their conjunction. Likewise, hiding a conditional secret is not equivalent to hiding the implication of the secret by the trigger. Thus, the two variants require an extension of the solution for the case of a single or unconditional secret.  In Remark~\ref{rem ms sol}, we describe such an extension.

\subsection{Automata and LTL}
\label{aut and ltl}
An \emph{automaton} on infinite words is $\A = \langle \Sigma, Q, q_0, \delta, \alpha  \rangle$, where $\Sigma$ is an alphabet, $Q$ is a finite set of \emph{states}, $q_0\in Q$ is an \emph{initial state}, $\delta: Q\times \Sigma \to 2^Q $ is a  \emph{transition function}, and $\alpha$ is an \emph{acceptance condition}, to be defined below. 
For states $q,s \in Q$ and a letter $\sigma \in \Sigma$, we say that $s$ is a $\sigma$-successor of $q$ if $s \in \delta(q,\sigma)$. 
Note that we do not require the transition function to be \emph{total}. 
That is, we allow that $\delta(q,\sigma)=\varnothing$.
If $|\delta(q, \sigma)| \leq 1$ for every state $q\in Q$ and letter $\sigma \in \Sigma$, then $\A$ is \emph{deterministic}. For a deterministic automaton $\A$ we view $\delta$ as a function $\delta:Q\times \Sigma\rightarrow Q\cup\set{\bot}$, where $\bot$ is a distinguished symbol, and instead of writing $\delta(q,\sigma)=\set{q}$ and $\delta(q,\sigma)=\varnothing$, we write $\delta(q,\sigma)=q$ and $\delta(q,\sigma)=\bot$, respectively.

A \emph{run}  of $\A$ on $w = \sigma_0 \cdot \sigma_1 \cdots \in \Sigma^\omega$ is an infinite sequence of states $r = r_0\cdot r_1\cdot r_2\cdot \ldots \in Q^\omega$, such that $r_0 = q_0$, and for all $i \geq 0$, we have that $r_{i+1} \in \delta(r_i, \sigma_{i})$.  The acceptance condition $\alpha$ determines which runs are ``good''. We consider here the \emph{B\"uchi}, \emph{co-B\"uchi}, \emph{generalized B\"uchi} and {\em parity\/} acceptance conditions. All conditions refer to the set  ${\it inf}(r)\subseteq Q$ of states that $r$ traverses infinitely often. Formally, ${\it inf}(r) = \{  q \in Q: q = r_i \text{ for infinitely many $i$'s}   \}$. 
In generalized B\"uchi the acceptance condition is of the form $\alpha=\{\alpha_1,\alpha_2,\ldots,\alpha_k\}$, for $k \geq 1$ and sets $\alpha_i \subseteq Q$. In a generalized B\"uchi automaton, a run $r$ is accepting if for all $1 \leq i \leq k$, we have that  ${\it inf}(r)\cap \alpha_i \neq \varnothing$. Thus, $r$ visits each of the sets in $\alpha$ infinitely often. B\"uchi automata is a special case of its generalized form with $k=1$. That is, a run $r$ is accepting with respect to the B\"uchi condition $\alpha\subseteq Q$, if ${\it inf}(r)\cap \alpha\neq \varnothing$. Dually, in co-B\"uchi automata, a run $r$ is accepting if ${\it inf}(r)\cap \alpha = \varnothing$. Finally, in a parity automaton, the acceptance condition $\alpha:Q\to \{1,...,k\}$, for some $k \geq 1$, maps states to ranks, and a run $r$ is accepting if the maximal rank of a state in ${\it inf}(r)$ is even. Formally, $\max_{q \in {\it inf}(r)} \{\alpha(q)\}$ is even. A run that is not accepting is \emph{rejecting}. We refer to the number $k$ in $\alpha$ as the {\em index\/} of the automaton.
A word $w$ is accepted by $\A$ if there is an accepting run of $\A$ on $w$. The language of $\A$, denoted $L(\A)$, is the set of words that $\A$ accepts. Two automata are \emph{equivalent} if their languages are equivalent. 

We denote the different classes of automata by three-letter acronyms in $\{ \text{D,N} \} \times \{ \text{B,C,GB,P}\} \times \{\text{W}\}$. The first letter stands for the branching mode of the automaton (deterministic, nondeterministic); the second for the acceptance condition type (B\"uchi, co-B\"uchi, generalized B\"uchi, or parity); and the third indicates we consider automata on words. 
For example, NBWs are nondeterministic B\"uchi word automata. 

\LTL is a linear temporal logic used for specifying on-going behaviors of reactive systems \cite{Pnu81}. 
Specifying the behavior of $(I/O)$-transducers, formulas of \LTL are defined over the set $I \cup O$ of signals using the usual
Boolean operators and the temporal operators $G$ (``always") and $F$ (``eventually"), $X$ (``next time'') and
$U$ (``until''). The semantics of \LTL is defined with respect to infinite 
computations in $(2^{I \cup O})^\omega$. Thus, each \LTL formula $\spec$ over $I \cup O$ induces a language $L_\spec\subseteq (2^{I \cup O})^\omega$ of all computations that satisfy $\spec$. 

Recall that the input to the problem of synthesis with privacy includes languages $L_\spec$ and $L_\sec$. We sometimes replace $L_\spec$ and $L_\sec$ in the different notations with automata or \LTL formulas that describe them, thus talk about realizability with privacy of $\zug{\A_\spec,\A_\sec,\H}$ or $\zug{\spec,\sec,\H}$, for automata $\A_\spec$ and $\A_\sec$, or \LTL formulas $\spec$ and $\sec$. 

\begin{exa}
\label{scheduler}
Consider a scheduler that serves two users and grant them with access to a joint resource. 
The scheduler can be viewed as an open system with $I=\{\req_1,\req_2\}$, with $\req_i$ ($i \in \{1,2\}$) standing for a request form User~$i$, and $O=\{\grant_1,\grant_2\}$, with $\grant_i$ standing for a grant to User~$i$. The system should satisfy mutual exclusion and non-starvation. Formally, the specification for the system is $\spec_1 \wedge \spec_2 \wedge \spec_3$, for $\spec_1=G((\neg \grant_1) \vee (\neg \grant_2))$, $\spec_2=G(\req_1 \rightarrow F \grant_1)$, and $\spec_3=G(\req_2 \rightarrow F \grant_2)$. 

We may want to hide from an observer of the interaction the exact policy scheduling of the system. For example,\footnote{The LTL operator $W$ is ``weak Until'', thus $p_1Wp_2=(p_1Up_2) \vee Gp_1$.} the secret 
$\sec_1=((\neg \grant_1) W \req_1) \wedge G(\grant_1 \rightarrow X((\neg \grant_1) W \req_1))$ reveals whether the system gives User~$1$ grants only after requests that have not been granted yet. Indeed, $\sec_1$ specifies that once a grant to User~$1$ is given, no more grants are given to her, unless a new request from her arrives. A similar secret can be specified for User~$2$. Note that in order to hide $\sec_1$, it is sufficient to hide only one of the signals $\req_1$ or $\grant_1$. In fact, this is true even when the observer knows the specification for the system. 
Then, the secret
$\sec_2=G((\req_1 \rightarrow \grant_1 \vee X \grant_1)\wedge (\req_2 \rightarrow \grant_2 \vee X\grant_2))$ reveals whether delays in grants are limited to one cycle. Here, unlike with $\sec_2$, it is not sufficient hiding only a single request or even both. Indeed, some policies disclose the satisfaction value of $\sec_2$ even when requests are hidden. For example, a system that simply alternates between grants, thus outputs $\set{\grant_1},\set{\grant_2},\set{\grant_1},\set{\grant_2},\ldots$, satisfies the specification and clearly satisfies $\sec_2$ regardless of the users' requests.

Consider now the secret $\sec_3=FG(\req_1 \rightarrow \grant_1)$, which asserts that eventually, the requests of User~$1$ are always granted immediately. A system that satisfies $\sec_3$ is unfair to User~$2$. Aiming to hide this unfair behavior, we can use the conditional secret $\zug{\sec_3,\sec_3}$, which requires a system that satisfies $\sec_3$ to hide its satisfaction. 

Some computations that satisfy $\sec_3$, however, may still be fair to User~$2$. For example, if $\sec_3$ is satisfied vacuously or if only finitely many requests are sent from User~$2$, then the behavior specified in $\sec_3$ is fair, and we need not hide it. Accordingly, we can strengthen the trigger $\sec_3$ and restrict further the computations in which the satisfaction value of $\sec_3$ should be hidden. Formally, we replace the trigger $\sec_3$ by a trigger $\sec_3 \wedge \trig$, for a behavior $\trig$ in which a scheduling policy that satisfies $\sec_3$ is not fair (and hence, need to be hidden).

Let us consider possible behaviors $\trig$ for the conditional secret $\zug{\sec_3 \wedge \trig,\sec_3}$. As discussed above, behaviors that make $\sec_3$ unfair are $GF \req_1$, implying that $\sec_3$ is not satisfied vacuously, and $GF \req_2$, implying that the immediate grants to User~$1$ are not due to no requests from User~$2$. 
Taking $\trig=(GF \req_1) \wedge (GF \req_2)$ results in a more precise conditional secret. 

The trigger $\trig$ can be made more precise: taking $\trig=GF (\req_1 \wedge \req_2)$ still guarantees no vacuous satisfaction and also asserts that immediate grants to User~$1$ are given even when the requests of User~$1$ arrive together with those of User~$2$. In fact, $\trig=GF(\req_2 \wedge (\neg \grant_2) U \req _1))$ is even more precise, as it excludes the possibility that the requests of User~$1$ arrive before those of User~$2$. Note that the secret can be made less restrictive too, for example with $\sec'_3=FG(\req_1 \rightarrow ((\neg \grant_2) U \grant_1))$, which specifies that eventually, grants to User~$1$ are always given before grants to User~$2$. 
\hfill \qed
\end{exa}

\begin{exa}
\label{robot}
As a different example, consider a paint robot that paints parts of manufactured pieces. The set $O$ includes $3$ signals $c_0,c_1,c_2$ that encode $8$ colors. The encoding corresponds to the composition of the color from paint in three different containers. For example, color $101$ stands for the robot mixing paints from containers $0$ and $2$. The observer does not see the generated pattern, but, unless we hide it, may see the arm of the robot when it reaches a  container. Accordingly, hiding of signals in $O$ involve different costs. 

The user instructs the robot whether to stay with the current color or change it, thus $I=\{\change\}$. We seek a system that directs the robot which color to choose, in a way that satisfies requirements about the generated pattern. 
For example, in addition to the requirement to respect the changing instructions ($\xi_{\it respect}$), the specification $\spec$ may require the pattern to start with color $000$, and if there are infinitely many changes, then all colors are used ($\xi_{\it all}$), yet color $000$ repeats between each two colors ($\xi_{\it repeat}$). Formally, 
\begin{itemize}
\item
$\xi_{\it respect}=G((X\neg \change) \leftrightarrow ((c_0 \leftrightarrow X c_0) \wedge (c_1 \leftrightarrow X c_1) \wedge (c_2 \leftrightarrow X c_2)))$,
\item
$\xi_{\it all}=GF (\bar{c_0}\wedge \bar{c_1}\wedge \bar{c_2}) \wedge   GF (\bar{c_0}\wedge \bar{c_1}\wedge c_2) \wedge \cdots \wedge GF(c_0 \wedge c_1 \wedge c_2)$,
\item
$\xi_{\it repeat}= G((c_0 \vee c_1 \vee c_2) \rightarrow X (\change \rightarrow (\bar{c_0}\wedge \bar{c_1}\wedge \bar{c_2})))$, and
\item
$\spec=(\bar{c_0}\wedge \bar{c_1}\wedge \bar{c_2}) \wedge \xi_{\it respect} \wedge ((GF \change) \rightarrow (\xi_{\it all} \wedge \xi_{\it repeat}))$. 
\end{itemize}
We may want to hide from an observer certain patterns that the robot may produce. For example, the fact color $111$ is used only after color $110$ (with color $000$ between them), the fact there are colors other than $000$ that repeat without a color different from $000$ between them, and more. Note that not all the signals in $O$ need to be hidden, and that the choice of signals to hide depends on the secrets as well as the cost of hiding. 
\hfill \qed
\end{exa}

\section{Solving Synthesis with Privacy}
\label{solve syn priv}
In this section we describe a solution to the problem of synthesis with privacy for \LTL specifications and show that it is TIME-complete, thus not harder than \LTL synthesis. The solution is based on replacing the specification by one that guarantees the hiding of the secret. For this, we need the following two constructions.

\begin{lem}\label{noise nxw}
	Consider a nondeterministic automaton $\A=\zug{2^{I \cup O},Q,q_0,\delta,\alpha}$. Given a set $\H\subseteq I\cup O$, there is a transition function $\delta^{\H}:Q \times 2^{I \cup O} \rightarrow 2^Q$ such that the nondeterministic automaton 
	$\A^{\H}=\zug{2^{I \cup O},Q,q_0,\delta^{\H},\alpha}$ is such that $L(\A^\H)=\noiseH(L(\A))$.
\end{lem}
 
 \begin{proof}
	Intuitively, the transition function $\delta^{\H}$ increases the nondeterminism of $\delta$ by guessing an assignment to the signals in $\H$. Formally, for $q \in Q$ and $\sigma \in 2^{I \cup O}$, we define $\delta^{\H}(q,\sigma) = \bigcup_{\sigma' \in \noiseH(\sigma)} \delta(q,\sigma')$. It is easy to see that a word $w'$ is accepted by $\A^{\H}$ iff there is a word $w$ accepted by $\A$ such that $w' \in \noiseH(w)$. 
\end{proof} 
 
Note that while $\A^{\H}$ maintains the state space and acceptance condition of $\A$, it does not preserve determinism. Indeed, unless $\H=\varnothing$, we have that  $\A^{\H}$ is nondeterministic even when $\A$ is deterministic. Next, in \autoref{dpw} we construct automata that accept computations that satisfy the specification and hide the secret when a given set of signals is hidden. 

\begin{lem}\label{dpw}
	Consider two disjoint finite sets $I$ and $O$, a subset $\H\subseteq I\cup O$, and $\omega$-regular languages $L_\spec$ and $L_\sec$ over the alphabet $2^{I \cup O}$. There exists a DPW $\D^\H_{\spec,\sec}$ with alphabet $2^{I\cup O}$ that accepts a computation $w\in (2^{I\cup O})^\omega$ iff $w \in L_\spec$ and there exist $w^+,w^-\in \noiseH(w)$ such that $w^+ \in L_\sec$ and $w^-\not\in L_\sec$. 
	\begin{enumerate}
		\item
		If $L_\spec$ and $L_\sec$ are given by \LTL formulas $\spec$ and $\sec$, then $\D^\H_{\spec,\sec}$ has $2^{2^{O(|\spec|+|\sec|)}}$ states and index $2^{O(|\spec|+|\sec|)}$. 
		\item
		If $L_\spec$ and $L_\sec$ are given by DPWs $\D_\spec$ and $\D_\sec$ with $n_\spec$ and $n_\sec$ states, and of indices $k_\spec$ and $k_\sec$, then $\D^\H_{\spec,\sec}$ has $2^{O(n_\spec\cdot k_\spec\cdot (n_\sec\cdot k_\sec)^2 \log (n_\spec\cdot k_\spec\cdot n_\sec\cdot k_\sec) )}$ states and index $O(n_\spec\cdot k_\spec\cdot (n_\sec\cdot k_\sec)^2)$.
	\end{enumerate}
\end{lem}

\begin{proof}
	We start with the case $L_\spec$ and $L_\sec$ are given by \LTL formulas $\spec$ and $\sec$.
	Let $\A_\spec$, $\A_\sec$ and $\A_{\lnot \sec}$ be NGBWs for $L_\spec$, $L_\sec$, and $L_{\neg \sec}$. By \cite{VW94}, such NGBWs exist, and are of size exponential in the corresponding \LTL formulas. 
	Let $\A_\sec^\H$ and $\A_{\lnot\sec}^\H$ be the NGBWs for $\noiseH(L(\A_\sec))$ and $\noiseH(L(\A{\lnot\sec}))$, respectively, constructed as in Lemma~\ref{noise nxw}. 
		
	Now, let $\N^\H_{\spec,\sec}$ be an NGBW for the intersection of the three automata $\A_\spec$, $\A^\H_\sec$, and $\A^\H_{\lnot \sec}$. The NGBW $\N$ can be easily defined on top of the product of the three automata, and hence is of size $2^{O(|\spec|+|\sec|)}$ and index $O(|\spec|+|\sec|)$. Observe that indeed, a word $w\in (2^{I\cup O})^\omega$ is accepted by $\N^\H_{\spec,\sec}$ iff $w\models  \spec$ and there exist $w^+,w^-\in \noiseH(w)$ such that $w^+\models \sec$ and $w^-\not\models \sec$. By \cite{Saf88,Pit06}, determinizing $\N^\H_{\spec,\sec}$ results in a DPW $\D^\H_{\spec,\sec}$ with $2^{2^{O(|\spec|+|\sec|)}}$ states and index $2^{O(|\spec|+|\sec|)}$, and we are done.
		
	We continue with the case $L_\spec$ and $L_\sec$ are given by DPWs $\D_\spec$ and $\D_\sec$. We first obtain from $\D_\sec$ two NBWs, $\A_\sec$ and $\A_{\neg\sec}$ for $L_\sec$ and $L_{\neg\sec}=(2^{I\cup O})^\omega\setminus L_\sec$ respectively, and also we translate $\D_\spec$ into an equivalent NBW $\A_\spec$. Note that the NBWs $\A_\sec$ and $\A_{\neg\sec}$ can be defined with $O(n_\sec\cdot k_\sec)$ states, and that $\A_\spec$ can be defined with $O(n_\spec\cdot k_\spec)$ states. We then obtain the NBWs $\A_\sec^\H$ and $\A_{\lnot\sec}^\H$ by applying the construction in Lemma~\ref{noise nxw} on $\A_{\sec}$ and $\A_{\neg \sec}$, respectively. We then define an NBW of size $O(n_\spec\cdot k_\spec\cdot (n_\sec\cdot k_\sec)^2)$ for the intersection of the three NBWs $\A_\spec$, $\A_\sec^\H$, and $\A_{\lnot\sec}^\H$, and finally determinize it into a DPW $\D^\H_{\spec,\sec}$ with $2^{O(n_\spec\cdot k_\spec\cdot (n_\sec k_\sec)^2 \log (n_\spec\cdot k_\spec\cdot n_\sec k_\sec))}$ states and index $O(n_\spec\cdot k_\spec\cdot (n_\sec k_\sec)^2)$.
\end{proof}

\begin{rem}[The size of $\D^\H_{\spec,\sec}$ for specifications and secrets given by DBWs]\label{when dbw}
The exponential dependency of $\D^\H_{\spec,\sec}$ in the DPW $\D_\spec$  in the construction in Lemma~\ref{dpw} follows from the exponential blow up in DPW intersection \cite{Bok18}.
When $L_\spec$ is given by a DBW $\D_\spec$, we can first construct a DPW for the intersection of $\A_\sec^\H$ and $\A_{\lnot\sec}^\H$,
and only then take its intersection with $\D_\spec$. This results in a DPW $\D^\H_{\spec,\sec}$ of size exponential in $\D_\sec$, but only polynomial in $\D_\spec$. 
\end{rem}
	
We can now solve synthesis with privacy for LTL formulas.

\begin{thm}[Synthesis with privacy, LTL]\label{2exp ltl}
	Given two disjoint finite sets $I$ and $O$, LTL formulas $\spec$ and $\sec$ over $I\cup O$, a cost function $\cost:I\cup O\rightarrow \Nat$, and a budget $b\in \Nat$, deciding whether $\zug{\spec,\sec,\cost,b}$ is realizable with privacy is 2EXPTIME-complete. 
\end{thm}
\begin{proof}
	We start with the upper bound. Given $\spec$, $\sec$, $\cost$, and $b$, we go over all $\H\subseteq I\cup O$ such that $\cost(\H)\leq b$, construct the DPW $\D^{\H}_{\spec,\sec}$ defined in Lemma~\ref{dpw}, and check whether $L(\D^{\H}_{\spec,\sec})$ is realizable. Since realizability of a DPW with $n$ states and index $k$ can be solved in time at most $O(n^k)$ \cite{CJKLS17}, the 2EXPTIME upper bound follows from $\D^{\H}_{\spec,\sec}$ having $2^{2^{O(|\spec|+|\sec|)}}$ states and index $2^{O(|\spec|+|\sec|)}$.
	
	For the lower bound, we describe a reduction from \LTL synthesis with no privacy. Note that adding to a specification $\spec$ a secret $\True$ or $\False$ does not work, as an observer knows its satisfaction value. It is easy, however, to add a secret that is independent of the specification. Specifically, given a specification $\spec$ over $I\cup O$, let $O'=O\cup \set{p}$, where $p$ is a fresh signal not in $I\cup O$. Consider the secret $\sec=p$ and a cost function with $\cost(p)=0$. Clearly, an $(I/O)$-transducer $\T$ realizes $\spec$ iff the $(I/O')$-transducer $\T'$ that agrees with $\T$ and always assigns $\False$ to $p$, realizes $\spec$ and $\set{p}$-hides $\sec$. Conversely, for an $I/O'$-transducer $\T'$, let $\T$ be the $(I/O)$-transducer obtained from $\T$ by ignoring the assignments to $p$. Clearly, $\T'$ $\set{p}$-hides $\sec$. In addition, as $\spec$ does not refer to $p$, we have that $\T'$ realizes $\spec$ iff $\T$ realizes $\spec$. Thus, $\spec$ is realizable iff $\zug{\spec,\sec,\cost,0}$ is realizable with privacy.
\end{proof}

\begin{rem}[Solving privacy with multiple and conditional secrets]\label{rem ms sol}
Recall that for a set of secrets $S=\set{\sec_1,\sec_2,\ldots,\sec_k}$, an $(I/O)$-transducer $\T$ realizes $\zug{\spec,S,\H}$ with privacy if it realizes $\spec$ and $\H$-hides $\sec_i$, for all $i\in [k]$. It is easy to extend Theorem~\ref{2exp ltl} to the setting of multiple secrets by replacing the DPW $\D^\H_{\spec,\sec}$ by a DPW obtained by determinizing the product of $\A_\spec$ with automata $\A^\H_{\sec_i}$ and $\A^\H_{\lnot \sec_i}$, for all $1 \leq i \leq k$. 

As for conditional secrets, recall that a computation should satisfy the specification, and from the point of view of an observer, either the trigger is not triggered, thus $\pi\in L(\A^\H_{\neg\trig})$, or the secret is hidden, thus $\pi \in L(\A^\H_\sec)\cap L(\A^\H_{\neg\sec})$. Accordingly, we need to construct a deterministic automaton for $L(\A_\spec)\cap(L(\A^\H_{\neg\trig})\cup L(\A^\H_\sec)\cup L(\A^\H_{\neg\sec}))$. This can done by determinizing an NBW that is defined on top of the product of $\A_\spec,\A^\H_{\neg\trig},\A^\H_\sec$, and $\A^\H_{\neg\sec}$.

\end{rem}

While the complexity of our algorithm is not higher than that of \LTL synthesis with no privacy, it would be misleading to state that handling of privacy involves no increase in the complexity. Indeed, the algorithm involved two components whose complexity may have been dominated by the doubly exponential translation of the \LTL formulas to deterministic automata:
\begin{enumerate}
	\item
	A need to go over all candidate sets $\H \subseteq I \cup O$.
	\item
	A need to check that the generated transducer $\H$-hides the secret. 
\end{enumerate}

In the next two sections, we isolate these two components of synthesis with privacy and show that each of them involves an exponential complexity: the first in the number of signals and the second in the size of the secret. 

\subsection{Hiding secrets is hard}\label{hiding hard}

The synthesis problem for DBWs can be solved in polynomial time. Indeed, the problem can be reduced to solving a B\"uchi game played on top of the specification automaton. In this section we show that synthesis with privacy is EXPTIME hard even for a given set $\H$ of hidden signals (in fact, even a singleton set $\H \subseteq I$), a trivial specification, and a secret given by a DBW. 

We start by showing that $\H$-hiding is hard even for secrets given by DBWs. 

\begin{thm}\label{lb dbw hiding}
	Given two disjoint finite sets $I$ and $O$, a DBW $\D_\sec$ over $2^{I\cup O}$, and a set $\H\subseteq I\cup O$ of hidden signals, deciding whether there exists an $(I/O)$-transducer that $\H$-hides $\D_\sec$ is EXPTIME-hard.
	The problem is EXPTIME-hard already when $\H\subseteq I$.
\end{thm}

\begin{proof}
	We describe a polynomial-time reduction from NBW realizability, which is EXPTIME-hard~\cite{Rab72,HKKM02}.
	Given an NBW $\A$ over $2^{I\cup O}$, we define a set of signals $\H$ and a DBW $\D_\sec$ over $2^{I\cup O \cup \H}$, such that $L(\A)$ is realizable iff there exists an $((I \cup \H) / O)$-transducer that $\H$-hides $L(\D_\sec)$.
	
	Let $\A=\zug{2^{I\cup O}, Q, q_0, \delta, \alpha}$. W.l.o.g, we assume that $\A$ has a single initial state and  that every word in $(2^{I\cup O})^{\omega}$ has at least one rejecting run in $\A$. The latter can be achieved, for example, by adding a nondeterministic transition from the initial state to a rejecting sink upon any assignment $i\cup o\in 2^{I\cup O}$.
	Let $\H$ be a set of signals that encode $Q$.
	Thus, each assignment $s \in 2^\H$ is associated with a single state in $Q$.
	We refer to a letter in $2^{I \cup O \cup \H}$ as a pair $\zug{\sigma,q} \in 2^{I \cup O} \times Q$, and
	we view a word in $(2^{I \cup O \cup \H})^\omega$ as the combination $w \oplus r$, of a word  $w \in (2^{I \cup O})^\omega$ with a word $r \in Q^\omega$. Formally, for $w=\sigma_0 \cdot \sigma_1 \cdots \in (2^{I \cup O})^\omega$ and $r = r_1 \cdot r_2 \cdots \in Q^\omega$, let $w \oplus r = \zug{\sigma_0,r_1} \cdot \zug{\sigma_1,r_2} \cdots \in (2^{I \cup O \cup \H})^\omega$. Then, 
	we define $\D_\psi$ so that
	$L(\D_\psi)=\{w \oplus r\in (2^{I \cup O \cup \H})^\omega :  \mbox{the sequence $q_0 \cdot r$ is an accepting run of $\A$ on } w\}$.
	Note that since every word in $(2^{I\cup O})^{\omega}$ has at least one rejecting run in $\A$, then every word $w \in L(\A)$ has at least one word $r^+ \in Q^\omega$ such that $w \oplus r^+ \in L(\D_\psi)$ and at least one word $r^- \in Q^\omega$ such that $w \oplus r^- \not \in L(\D_\psi)$.
	
	Formally, $\D_\psi=\zug{2^{I \cup O \cup \H},Q,q_0,\delta',\alpha}$ has the same state space and acceptance condition as $\A$, and it uses the $Q$-component of each letter in order to resolve the nondeterministic choices in $\A$. Thus, the transitions function $\delta':Q \times 2^{I \cup O\cup \H} \rightarrow Q$ is defined as follows. For every state $q \in Q$ and letter $\zug{\sigma,s} \in 2^{I \cup O \cup \H}$, we have that $\delta'(q,\zug{\sigma,s})=s$ if $s\in \delta(q,\sigma)$, and otherwise $\delta'(q,\zug{\sigma,s})=\bot$.

	We prove that indeed $L(\D_\psi)$ accepts exactly all words $w \oplus r$ such that $q_0 \cdot r$ is an accepting run of $\A$ on $w$. 
	By definition of $\delta'$, it holds that $r'$ is a run of $\D_\sec$ over $w\oplus r$ iff $r'=q_0\cdot r$, and $q_0\cdot r$ is a run of $\A$ over $w$. Hence, a run $r'=q_0\cdot r$ of $\D_\sec$ over $w\oplus r$ is accepting, iff ${\it inf}(r')\cap \alpha\neq \varnothing$, iff $r'=q_0\cdot r$ is an accepting run of $\A$ over $w$, and we are done. 
\end{proof}

Note that in the proof of Theorem~\ref{lb dbw hiding}, we could have defined $\H$ so that it resolves the nondeterminism in $\A$ in a more concise way. In particular, if we assume that the nondeterminism degree in $\A$ is at most $2$, then a set $\H$ of size $1$ can resolve the nondeterminizm of $\delta$. Hence, as NBW synthesis is EXPTIME-hard already for NBWs with branching degree~$2$ (this follows from the fact that a bigger branching degree can be decomposed along several transitions), EXPTIME hardness holds already when hiding a single input signal. 

\begin{thm}[Synthesis with privacy, DPWs]\label{exp dbw}
	Given two disjoint finite sets $I$ and $O$, DPWs $\D_\spec$ and $\D_\sec$ over $2^{I\cup O}$, and a set $\H\subseteq I\cup O$ of hidden signals, deciding whether $\zug{\D_\spec,\D_\sec,\H}$ is realizable with privacy is EXPTIME-complete. Moreover, hardness holds already when the specification is trivial and the secret is given by a DBW.
\end{thm}

\begin{proof}
	For the upper bound, we solve the synthesis problem for the DPW $\D^{\H}_{\spec,\sec}$ defined in Lemma~\ref{dpw}. 
	As specified there, the size of $\D^{\H}_{\spec,\sec}$ is exponential in the size and index of both $\D_\spec$ and $\D_\sec$, and its index is polynomial in the size and index of $\D_\spec$ and $\D_\sec$. Membership in EXPTIME then follows from the complexity of the synthesis problem for DPWs \cite{BCJ18}.
	
	For the lower bound, fix a DBW $\D_\True$ such that $L(\D_\True)=(2^{I\cup O})^\omega$. Then, it is easy to see that $\zug{\D_\True,\D_\sec,\H}$ is realizable with privacy iff there is an $(I/O)$-transducer that $\H$-hides $\D_\sec$. Thus, hardness in EXPTIME follows from Theorem~\ref{lb dbw hiding}.
\end{proof}

\subsection{Searching for a set of signals to hide is hard}

Another component in the algorithm that is dominated by the doubly-exponential translation of \LTL to DPWs is the need to go over all subsets of $I \cup O$ in a search for the set $\H$ of signals to hide. Trying to isolate the influence of this search, it is not enough to consider specifications and secrets that are given by DBWs, as synthesis with privacy is EXPTIME-hard already for a given set $\H$, and so again, the complexity of the search is dominated by the complexity of the synthesis problem. Fixing the size of the secret, which is the source of the exponential complexity, does not work either, as it also fixes the number of signals that we may need to hide. We address this challenge by moving to an even simpler setting for the problem, namely synthesis with privacy of a {\em closed\/} system. We are going to show that in this setting, the search for $\H$ is the only non-polynomial component in the algorithm.

In the closed setting, all signals are controlled by the system, namely $I=\varnothing$. 
Consequently, each transducer has a single computation,  and realizability coincides with satisfiability.
In particular, for $I=\varnothing$, we have that $\zug{L_\spec,L_\sec,\H}$ is realizable with privacy iff there exists a word $w\in L_\spec$, for which there exist two words $w^+,w^-\in\noiseH(w)$ such that $w^+\in L_\sec$ and $w^-\notin L_\sec$. We show that while synthesis with privacy in the closed setting can be  solved in polynomial time for  a given set $\H$ of hidden signals, it is NP-complete when $\H$ is not given, even when the function ${\it cost}$ is uniform.

We start with the case $\H$ is given.
\begin{thm}\label{ptime}
	Given a finite set $O$ of output signals, a set $\H\subseteq O$ of hidden signals, and DPWs $\D_{\spec}$ and $\D_\sec$ over $2^O$, deciding whether $\zug{\D_\spec,\D_\sec,\H}$ is realizable with privacy can be done in polynomial time.
\end{thm}

\begin{proof}
	First, we complement $\D_{\sec}$, which results in a DPW $\D_{\neg\sec}$ of the same size, and of index $k+1$, where $k$ is the index of $\D_\sec$. Then, we
	translate $\D_\spec$, $\D_\sec$ and $\D_{\neg\sec}$ into equivalent NBWs $\A_\spec$, $\A_{\sec}$ and $\A_{\neg\sec}$, respectively. All three NBWs can be defined in size that is polynomial in their deterministic DPW counterpart.
	Let $\A_\sec^\H$ and $\A_{\lnot\sec}^\H$ be NBWs obtained by applying the construction in Lemma~\ref{noise nxw} on $\A_{\sec}$ and $\A_{\neg \sec}$, respectively. 
	By Lemma~\ref{noise nxw}, the NBWs $\A_\sec^\H$ and $\A_{\lnot\sec}^\H$ have the same number of states as $\A_{\sec}$ and $\A_{\neg \sec}$ respectively.
	Let $\N$ be an NBW for the intersection $L(\A_\spec)\cap L(\A_\sec^{\H})\cap L(\A_{\lnot \sec}^{\H})$. 
	Note that $\N$ can be defined with size that is polynomial in $\A_\spec$, $\A^\H_\sec$ and $\A^\H_{\neg\sec}$. Moreover, $\N$ accepts a word $w$ iff $w\in L(\A_\spec)$, and there exist two words $w^+,w^-\in\noiseH(w)$ such that $w^+\in L(\A_\sec)$ and $w^-\notin L(\A_\sec)$. Thus, realizability with privacy of $\zug{\A_\spec,\A_\sec,\H}$ can be reduced to the nonemptiness of $\N$, which can be decided in polynomial time. 
	\end{proof}

We continue to the case $\H$ should be searched.

\begin{thm}\label{np}
	Given a finite set of output signals $O$, DPWs $\A_\spec$ and $\A_\sec$ over $2^O$, a hiding cost function $\cost:O\rightarrow \Nat$, and a budget $b\in \Nat$, deciding whether $\zug{\A_\spec,\A_\sec,\cost,b}$ is realizable with privacy is NP-complete.
	Moreover, hardness holds already when the specification and secret are given by DBWs.
\end{thm}

\begin{proof}
	For the upper bound, a nondeterministic Turing machine can guess a set $\H\subseteq O$, check whether $\cost(\H)\leq b$, and, by Theorem~\ref{ptime}, check in polynomial time whether $\zug{\A_\spec,\A_\sec,\H}$ is realizable with privacy.
	
	For the lower bound, we describe a polynomial-time reduction from the {\em vertex-cover\/} problem.
	In this problem, we are given an undirected graph $G=\zug{V,E}$ and $k \geq 1$, and have to decide whether there is a set $S \subseteq V$ such that $|S| \leq k$ and for every edge $\{v,u\} \in E$, we have that $E \cap S \neq \varnothing$. 
	Given an undirected graph $G=\zug{V,E}$, with $E=\set{e_1,e_2,\ldots,e_m}$, we consider a closed setting with $O=V$ and construct DBWs $\A_\spec$ and $\A_\sec$ over the alphabet $2^V$ such that for all $\H\subseteq V$, it holds that $\zug{\A_\spec,\A_\sec,\H}$ is realizable with privacy iff $\H$ is a vertex cover of $G$. Accordingly,  there is a vertex cover of size $k$ in $G$ iff $\zug{\A_\spec,\A_\sec,\cost,k}$ is realizable with privacy for the uniform cost function that assigns $1$ to all signals in $O$. 
	
	We define $\A_\spec$ and $\A_\sec$ over the alphabet $2^V$ as follows.
	The DBW $\A_\spec$ is a 2-state DBW that accepts the single word ${\varnothing}^{\omega}$.
	The DBW $\A_\sec=\zug{2^V,Q,q_1,\delta,\alpha}$ for the secret is defined as follows. 
	The set of states is $Q=\set{q_1,q_2,\ldots,q_{m+1}}$. The set of accepting states is $\alpha=\set{q_{m+1}}$. The transition function $\delta$ is defined for all $S\subseteq O$ and $i\leq m$ as follows. If $S\cap e_i\neq \varnothing$, then $\delta(q_i,S)=q_{i+1}$, and otherwise, $\delta(q_i,S)=\bot$. Finally, $\delta(q_{m+1},S)=q_{m+1}$, for all $S\subseteq V$.
	
	Intuitively, words in $L(\A_\sec)$ encode vertex covers of $G$. Indeed, if $w=S_1\cdot S_2\cdot S_2\cdot \ldots \in L(\A_\sec)$, then for all $i\leq m$, we have that $S_i\cap e_i\neq \varnothing$. Thus, if for all $i\leq m$ we set $v_i\in V$ to be some vertex in $S_i\cap e_i$, then $\set{v_i,\ldots ,v_m}$ is a vertex cover of $G$. 
	
Formally, we prove that $\zug{\A_\spec,\A_\sec,\H}$ is realizable with privacy iff $\H$ is a vertex cover of $G$. 
Consider first a vertex cover $\H$ of $G$. 
Let $\T$ be a transducer that generates the computation ${\varnothing}^{\omega}$ and hides the signals in $\H$.
In every step of the computation, the observer sees that the signals in $V\setminus \H$ are assigned $\False$.
Assigning $\False$ to all the signals in $\H$ results in a word $w^-$ that is not in $L(\A_\sec)$.
On the other hand, since $\H$ is a vertex cover, i.e., it intersects every $e\in E$, assigning $\True$ to all the signals in $\H$ results in a word $w^+$ that is in $L(\A_\sec)$. 
Thus, $\T$ realizes $\zug{\A_\spec,\A_\sec,\H}$ with privacy.

For the other direction, consider a set $\H \subseteq V$ that is not a vertex cover of $G$.
Consider a transducer $\T$ that realizes $L(\A_\spec)$ and hides the set $\H$ of signals.
Since $\T$ realizes $L(\A_\spec)$, it generates the computation $\varnothing^{\omega}$.
Since $\H$ is not a vertex cover of $G$, there is $0\leq i\leq m$ such that $\H\cap \{e_i\}=\varnothing$.
Then, at every step of the computation, the observer sees the values of the signals in $e_i$.
In particular, at the $i$-th step of the computation, the observer sees that the signals in $e_i$ are assigned $\False$, which reveals that the computation is not in $L(\A_\sec)$. Thus, $\T$ does not realize $\zug{\A_\spec,\A_\sec,\H}$ with privacy, and we are done.	
\end{proof}	

\section{Bounded Synthesis with Privacy}
In the general synthesis problem, there is no bound on the size of the generated system. It is not hard to see that if a system that realizes the specification exists, then there is also one whose size is bounded by the size of a deterministic automaton for the specification. For the case of LTL specifications, this gives a doubly-exponential bound on the size of the generated transducer, which is known to be tight \cite{Ros92}. In \cite{SF07}, the authors suggested to study {\em bounded synthesis}, where the input to the problem includes also a bound on the size of the system. The bound not only guarantees the generation of a small system, if it exists, but also reduces the complexity of the synthesis problem and gives rise to a symbolic implementation and further extensions \cite{Ehl10,KLVY11}. In particular, for LTL, it is easy to see that bounded synthesis can be solved in PSPACE, as one can go over and model-check all candidate systems.
For specifications in DPW, the bound actually increases the complexity, as going over all candidates results in an algorithm in NP.

In this section we study bounded synthesis with privacy. As in traditional synthesis, the hope is to both reduce the complexity of the problem and to end up with smaller systems. In addition to a specification $L_\spec$, a secret $L_\sec$, and a  set  $\H\subseteq I\cup O$ of hidden signals, we are given a bound $n\in \Nat$, represented in unary, and we are asked to construct an $(I/O)$-transducer with at most $n$ states that realizes $\zug{L_\spec,L_\sec,\H}$ with privacy, or to determine that no such transducer exists.  As in the unbounded case, we can define the problem also with respect to a hiding cost function and a budget. 

\subsection{Hiding secrets by a bounded system  is hard}\label{bdd hiding hard subsec}
We first show that hiding secrets in a bounded setting is hard. In fact, the complexity of hiding goes beyond the complexity of bounded synthesis with no privacy already in the case the specification and secrets are given by LTL formulas. In the case of DBWs and DPWs, hiding is also more complex than bounded synthesis without privacy, but the difference is not significant. 

The key idea in both results is similar to the one in the proof of \autoref{lb dbw hiding}. There, we reduce realizability of NBWs to hiding of secrets given by DBWs. Essentially, we use the hidden signals to imitate nondeterminism. 
Here, with a bound on the size of the system, we cannot reduce from realizability, as the problem has the flavor of model checking many candidates. Accordingly, we reduce from universality, either in the form of LTL formulas with universally-quantified atomic propositions, or in the form of language-universality for NBWs. 

\begin{thm}\label{ltl hiding expspace}
	Given two disjoint finite sets $I$ and $O$, an LTL formula $\sec$ over $2^{I\cup O}$, a set $\H\subseteq I\cup O$ of hidden signals, and a bound $n\geq 1$, given in unary, deciding whether there exists an $(I/O)$-transducer of size at most $n$ that $\H$-hides $\sec$ is EXPSPACE-hard. The problem is EXPSPACE-hard already when $n=1$.
\end{thm}
\begin{proof}
	We describe a reduction from satisfiability of AQLTL, namely LTL with universally-quantified atomic propositions, which is known to be EXPSPACE-complete \cite{SVW87}. An AQLTL formula is of the form $\forall p_1,p_2,\ldots,p_k \theta$, where $\theta$ is an LTL formula over the set of atomic propositions $AP$, and $A=\set{p_1,p_2,\ldots,p_k}\subseteq AP$. The satisfiability problem is to determine whether there is $w\in (2^{AP})^\omega$ such that for all $w' \in \noise{A}(w)$, we have that $w' \models \theta$.

	Consider an AQLTL formula $\forall p_1,p_2,\ldots,p_k \theta$. Let  $A=\set{p_1,p_2,\ldots ,p_k}\subseteq AP$. 
	Let $I=AP\setminus A$ and $O=\H=A\cup \set{p}$, where $p$ is a fresh signal not in $AP$. Also, let $\sec=\theta\lor p$. 
	We claim that $\forall p_1,p_2,\ldots,p_k \theta$ is not satisfiable iff there is a single state $(I/O)$-transducer that $\H$-hides $\sec$. 
	
	Assume first that there is a single state $(I/O)$-transducer $\T$ that $\H$-hides $\sec$. Then, for every $w_I \in (2^{AP \setminus A})^{\omega}$, there is $w' \in \noise{A}(\T(w_I))$ such that $w' \not \models \sec$. In particular $w' \not \models \theta$. 
	Thus, $\forall p_1,p_2,\ldots,p_k \theta$ is not satisfiable. 
	
	Assume now that $\forall p_1,p_2,\ldots,p_k \theta$ is not satisfiable. Thus, for every $w\in (2^{AP})^\omega$, there is $w' \in \noise{A}(w)$ such that $w' \not \models \theta$. Also, as $p\in \H$, there are $w^+ \in \noise{\{p\}}(w')$ and $w^- \in \noise{\{p\}}(w')$ such that $w^+ \models p$ and $w^- \not \models p$. Note that $w^+,w^- \in \noise{A}(w)$ and that $w^+ \models \sec$ whereas $w^- \not \models \sec$. It follows that all $(I/O)$-transducers $\H$-hides $\sec$, and we are done.
\end{proof}

\begin{thm}\label{dbw hiding pspace}
	Given two disjoint finite sets $I$ and $O$, a DBW $\D_\sec$ over $2^{I\cup O}$, a set $\H\subseteq I\cup O$ of hidden signals, and a bound $n\geq 1$, represented in unary, deciding whether there exists an $(I/O)$-transducer of size at most $n$ that $\H$-hides $\D_\sec$ is PSPACE-hard.
	The problem is PSPACE-hard already when $n=1$.
\end{thm}
\begin{proof}
We describe a reduction from the NBW universality problem, which is PSPACE-hard \cite{SVW87}. Let $\N=\zug{2^I,Q,q_0,\delta,\alpha}$ be an NBW and w.l.o.g assume that for all $w\in (2^I)^\omega$, there is a rejecting run of $\N$ over $w$. Let $O=\H$ be a set of output signals that encode the states in $Q$. 
For convenience, we abuse notation and refer to elements in $2^O$ as states $q \in Q$. Let $\D_\psi=\zug{2^{I\cup O},Q,q_0,\delta',\alpha}$ be the DBW over the alphabet $2^{I\cup O}$, where the deterministic transition function $\delta':Q\times 2^{I\cup O}\rightarrow Q$ is defined for all $q\in Q$, $o\in 2^O$ and $i\in 2^{I}$, by $\delta'(q,i\cup o)=o$ if $o\in \delta(q,i)$, and $\delta'(q,i\cup o)=\bot$ otherwise. Clearly, the size of $\D_\psi$ is linear in that of $\N$. Also, by its  definition, we have that $\D_\psi$ accepts a word $\pi\in (2^{I\cup O})^\omega$ iff $r=q_0\cdot \hide_{I}(\pi)$ is an accepting run of $\N$ over the word $w_I=\hideH(\pi)$.

Let $o\in 2^O$ be some arbitrary output assignment, and consider the single state $(I/O)$-transducer $\T_o=\zug{I,O,\set{s_0},s_0,\eta,\tau}$, where $\eta(s_0,i)=s_0$ for all $i \in 2^I$ and $\tau(s_0)=o$. We claim that $\N$ is universal iff $\T_o$ $\H$-hides $\D_\psi$. 

First, as every word has a rejecting run in $\N$, then $\N$ is universal iff for all $w_I\in(2^I)^\omega$, there exist $\pi^+,\pi^-\in \noiseH(\T(w_I))$ such that $r^+=q_0\cdot\hide_{I}(\pi^+)$ is an accepting run of $\N$ over $w_I$, and $r^-=q_0\cdot \hide_{I}(\pi^-)$ is a rejecting run of $\N$ over $w_I$. Thus, $\pi^+\in L(\D_\psi)$ and $\pi^-\notin L(\D_\psi)$. It follows that $\N$ is universal iff $\T_o$ $\H$-hides $\D_\psi$. Since $o\in 2^O$ is arbitrary, we conclude that $\N$ is universal iff there exists a single state $(I/O)$-transducer $\T$ that $\H$-hides $\D_\psi$.
\end{proof}

\subsection{Solving bounded synthesis with privacy}
We can now present the tight complexity for bounded synthesis with privacy for both types of specification formalisms. 
For the upper bounds, we construct an NGBW $\N^\H_{\spec,\sec}$ that accepts exactly all words that satisfy $\spec$ and hide $\sec$, and search for an $(I/O)$-transducer of size $n$ whose language is contained in that of $\N^\H_{\spec,\sec}$. 

\begin{thm}[Bounded synthesis with privacy]\label{bounded syn complexity}
	Given two disjoint finite sets $I$ and $O$, specification $L_\spec$ and secret $L_\sec$ over $I\cup O$,  a set $\H\subseteq I\cup O$ of hidden signals, and a bound $n \in \Nat$, represented in unary, deciding whether there is 
an $(I/O)$-transducer with at most $n$ states that realizes $\zug{L_\spec,L_\sec,\H}$ with privacy is PSPACE-complete for $L_\spec$ and $L_\sec$  given by DPWs, and is EXPSPACE-complete for $L_\spec$ and $L_\sec$ given by LTL formulas. Hardness in PSPACE holds already for DBWs.
\end{thm}
\begin{proof}
	We start with the upper bound for the case $L_\spec$ and $L_\sec$ are given by DPWs $\D_\spec$ and $\D_\sec$. 
	As described in Lemma~\ref{dpw}, given $\D_\spec$ and $\D_\sec$, it is possible to construct an NGBW $\N^\H_{\spec,\sec}$ that accepts exactly all words $w\in L(\D_\spec)$ for which there exists $w^+,w^-\in \noiseH(w)$ with $w^+\in L(\D_\sec)$ and $w^-\notin L(\D_\sec)$. Moreover, $\N^\H_{\spec,\sec}$ is of constant index $k=3$, and of size polynomial in the size and index of $\D_\spec$ and $\D_\sec$.
	
	Once we have constructed $\N^\H_{\spec,\sec}$, we can go through all $(I/O)$-transducers of size $n$ and to check if for one of them it holds that all its computations satisfy $\N^\H_{\spec,\sec}$. This reduces to checking the containment of the language of each $(I/O)$-transducer of size $n$ in that of $\N^\H_{\spec,\sec}$, which can be done in PSPACE~\cite{SVW87}. 
	
	When we start with three NBWs $\A_\spec$, $\A_\sec$, and $\A_{\lnot\sec}$, for $L_\spec$, $L_\sec$, and $L_{\neg \sec}$, the construction of $\N^\H_{\spec,\sec}$ is polynomial. Hence, an EXPSPACE upper bound for LTL follows from the exponential translation of LTL formulas to NBWs. 
	
	For the lower bounds, note that  by taking $L_\spec=(2^{I\cup O})^\omega$ (that is, $\True$ for an LTL formula and the universal $\D_\True$ for DBW), we have that there is a transducer with $n$ states that realizes $\zug{L_\spec,L_\sec,\H}$ with privacy iff there is a transducer with $n$ states  that $\H$-hides $L_\sec$. Hence, the lower bounds follow from Theorems~\ref{ltl hiding expspace} and \ref{dbw hiding pspace}.  
\end{proof}

\section{Privacy with Certificates}
In this section we study the problem of synthesizing systems that not only maintain privacy against an observer in the usual sense, but also generate for the user a \emph{certificate for privacy}. This certificate is a computation that agrees with the real computation on the visible signals, yet it disagrees with it on the hidden signals in a way that induces a different satisfaction value for the secret. Accordingly, the certificate  witnesses that the secret is indeed hidden from the observer. 

Recall that a secret $\sec$ is $\H$-hidden in a computation $w\in (2^{I\cup O})^\omega$ iff there exist two computations $w^+,w^-\in \noiseH(w)$, such that $w^+\models \sec$ and $w^-\models \neg\sec$. 
Clearly, if $\sec$ is indeed $\H$-hidden in $w$, then we can always use $w$ as one of the witnesses $w^+$ or $w^-$, according to whether $\sec$ is satisfied or not in $w$. 
Hence,  $\sec$ is $\H$-hidden in $w$ iff there exists $w'\in \noiseH(w)$ such that ($w\models\sec$ iff $w'\models\neg\sec$) holds. 
Moreover, as all words in $\noiseH(w)$ agree with $w$ on $\V$, the witness word $w'$ is identified by its restriction $w'|_\H$ to the signals in $\H$. 
Accordingly, we may simply require the existence of a word $w_\H\in (2^\H)^\omega$, such that the witness word $w'\in \noiseH(w)$ for $\sec$ being $\H$-hidden in $w$ is such that $w'|_\H=w_\H$.
Thus, $w_\H$ describes a sequence of assignments to the signals in $\H$ with which the value of the secret $\sec$ is different from its value in $w$. 
We call $w_\H$ a \emph{certificate for the privacy of $\sec$ in $w$}. When $\sec$ and $w$ are clear from the context we simply refer to $w_\H$ as a \emph{certificate for privacy}. 

We can thus reformulate the notion of privacy against an observer as follows. An $(I/O)$-transducer $\T$ is said to $\H$-hide $\sec$ if for all $w_I\in (2^I)^\omega$, there exists a certificate $w_\H\in (2^\H)^\omega$ for $\sec$ in $\T(w_I)$. 
Note that a transducer $\T$ that $\H$-hides $\psi$ does not have to provide a certificate for $\psi$. An observer that tries to reveal the value of $\psi$ is not going to succeed, and the user is guaranteed that such a certificate exists,  but there is no requirement for $\T$ to generate it. Adding a certification may, however, be of interest: it assures the user that $\psi$ is indeed hidden. In this section we study the problem of synthesizing a system that not only realizes $\zug{\spec,\sec,\H}$ in the usual sense, but also generates for every input word an appropriate certificate. As we discuss below, adding certificates makes specifications and secrets more difficult to realize. That is, there are specification and secrets that become unrealizable with privacy once requiring the generation of certificates. Essentially, this follows from the fact that the certificate $w_\H$ has to be generated in an on-line manner.

Before we formally define the setting of \emph{certifiable privacy}, note that the complication added to the problem when requiring the online generation of a certificate is analogous to the difficulty in usual synthesis (with no privacy considerations) that is added on top of \emph{universal satisfiability}. In universal satisfiability we ask whether a specification is such that for all $w_I=i_0\cdot i_1\cdot i_2\cdots \in (2^I)^\omega$, there exists $w_O=o_0\cdot o_1\cdot o_2\cdots \in (2^O)^\omega$, such that $w_I\oplus w_O=(i_0\cup o_0)\cdot(i_1\cup o_1)\cdot(i_2\cup o_2)\cdots$ satisfies the specification. Then, in the synthesis problem, the inputs in $w_I$ are read one after the other, and we require that the output assignments in $w_O$ are generated online. 
Privacy (with no certificates) has the flavor of universal satisfiability with respect to the secret, while the setting with certification is closer to realizability as it is no longer enough that a certificate exists, instead we need to construct one on-line. It is hence not surprising that generating certificates makes the task of preserving privacy more difficult to realize. Indeed, we already know this is the case for realizability and universal satisfiability: There are specifications that are universally satisfiable yet not realizable. A simple example for such a specification is $\spec=o\leftrightarrow Xi$, which is clearly universally satisfiable, but cannot be realized as the system needs to predict the value of the input $i$ in order to generate the correct value for $o$ in the present. In \autoref{certificate example} we use similar formula to demonstrate the gap between synthesis with privacy, with and without certificates.

We continue and formally define the setting for synthesis with \emph{certified privacy}. 
A \emph{certifying transducer} $\C=\zug{I,O,\H,S,s_0,\rho,\tau,c}$ is defined similarly to a standard $(I/O)$-transducer but with a distinguished set of hidden signals $\H$ and an additional labeling function $c:S\rightarrow 2^\H$. 
For every input word $w_I=i_0\cdot i_1\cdot i_2\cdots \in (2^I)^\omega$, the run $r(w_I)=s_0\cdot s_1\cdot s_2\cdots \in S^\omega$, and the computation $\C(w_I)= (i_0\cup \tau(s_0))\cdot (i_1\cup \tau(s_1))\cdot (i_2\cup \tau(s_2))\cdots \in (2^{I\cup O})^\omega$, are defined as in a standard $(I/O)$-transducer. Then, the \emph{certificate} $c(w_I)\in (2^\H)^\omega$, for $\C(w_I)$, is defined to be $c(w_I)=c(s_0)\cdot c(s_1)\cdot c(s_2)\cdots $. Finally, the \emph{alternative computation} $\C'(w_I)\in (2^{I\cup O})^\omega$, is the unique computation that agrees with $\C(w_I)$ on the visible signals and with $c(w_I)$ on the hidden signals. Formally, $\C'(w_I)|_\V=\C(w_I)|_V$ and $\C'(w_I)|_\H=c(w_I)$.

Given LTL formulas $\spec$ and $\sec$ over $I\cup O$, and a set $\H\subseteq I\cup O$ of hidden signals, we say that $\C$ realizes $\zug{\spec,\sec,\H}$ with \emph{certified privacy}, if $\C$ realizes $\spec$ and for all $w_I\in (2^I)^\omega$, it holds that $c(w_I)$ is a certificate for $\sec$ being $\H$-hidden in $\C(w_I)$. I.e., for all $w_I\in (2^I)^\omega$ it holds that $\C(w_I)\models\spec$ and ($\C(w_I)\models \sec$ iff $\C'(w_I)\models\neg\sec$). 

Finally, in the \emph{synthesis with certified privacy problem}, we are given $\zug{\spec,\sec,\H}$, and we need to construct a certifying transducer $\C$ that realizes $\zug{\spec,\sec,\H}$ with certified privacy, or determine that no such transducer exists. 
\begin{exa}\label{certificate example}
	Let $I=\set{i}$, $\H=O=\set{o}$, and $\sec=G(o\leftrightarrow Xi)$. As we explain below, $\sec$ is $\H$-hidden in every $(I/O)$-transducer, yet there exists no certifying transducer that realizes $\zug{\True,\sec,\H}$ with certified privacy.  
	
	Consider a word $w=(i_0\cup o_0)\cdot (i_1\cup o_1)\cdot (i_2\cup o_2)\cdots \in (2^I)^\omega$. Clearly, $w\models\sec$ iff for all $j\geq 0$ we have that $o_j=\emptyset$ iff $i_{j+1}=\emptyset$. 
	This implies that $\sec$ is $\H$-hidden in $(I/O)$-transducers. Indeed, consider a transducer $\T$, then for every input word $w_I\in (2^I)^\omega$ there exist a unique word $w^+\in \noiseH(\T(w_I))$ that satisfies the aforementioned criterion, and thus satisfy $\sec$, while all other words in $\noiseH(\T(w_I))$ violate $\sec$.
	Note that then it follows that a certifying transducer $\C$ should manage to generate the exact unique $w^+$ either as the real computation $\C(w_I)$ or as the alternative computation $\C'(w_I)$. We show that this is impossible. Specifically, that the environment can force both computations to violate $\sec$. 
	
	Consider a certifying transducer $\C$. We prove that there exists an input word $w_I\in (2^I)^\omega$ such that both $\C(w_I)$ and $\C'(w_I)$ do not satisfy $\sec$.
	The word $w_I$ is defined so that the criterion $i_{j+1}=\emptyset$ iff $o_j=\emptyset$ is violated in $\C(w_I)$ at time $j=0$ and in $\C'(w_I)$ at time $j=1$. Thus, after two time steps both the real computation and the alternative computation violate the secret $\sec$, implying that $\C$ does not realize $\zug{\True,\sec,\H}$ with certified privacy.
	
	First, we define $i_0=\emptyset$. Then, let $o_0\in 2^O$ and $o'_0\in 2^O$ be the first real output of $\C$ and the first certificate output of $\C$, when reading the first input $i_0$. I.e., if $s_1=\rho(s_0,i_0)$, then $o_0=\tau(s_1)$ and $o'_0=c(s_1)$. 
	We choose $i_1\in 2^I$ so that $i_1=\emptyset$ iff $o_0\neq \emptyset$ holds. Observe that now the computation of $\C$ already violates $\sec$. Let now $o_1$ and $o'_1$ be the real and certificate outputs of $\C$ after reading the prefix $i_0\cdot i_1$. Similar to the previous choice of $i_1$, we now choose $i_2$ so that $i_2=\emptyset$ iff $o'_1\neq\emptyset$. As before, we can see that now the alternative computation generated by $\C$ is guaranteed to violate $\sec$ as well. Thus, no matter how we extend $i_0\cdot i_1\cdot i_2$ into an infinite word $w_I$, it holds that both $\C(w_I)$ and $\C'(w_I)$ violate $\sec$. In particular, if we let $w_I=i_0\cdot i_1\cdot i_2\cdot \emptyset^\omega$, then $\C(w_I)\models\neg\sec$ and $\C'(w_I)\models\neg\sec$. That is, $\C$ does not certify that $\sec$ is $\H$-hidden and we are done. 
	\hfill \qed
\end{exa}

\subsection{Solving certified privacy}
We continue to the solution of the synthesis problem with certified privacy. As discussed above, the need to generate the certificate on-line makes some specifications and secrets non-realizable. On the positive side, it enables the use of universal co-B\"uchi automata, which enable a Safraless synthesis procedure \cite{KV05c}. Thus, the same way universal satisfiability is used as a heuristic to solve the realizability problem \cite{KSS11}, our procedure here can be viewed as a heuristic to the problem of synthesis with privacy (without a certificate). 

\begin{lem}\label{dpw with certificates}
	Consider two disjoint finite sets $I$ and $O$, a subset $\H\subseteq I\cup O$, and LTL formula $\spec$ and $\sec$ over the signals $I \cup O$. There exists a UCW $\U^\H_{\spec,\sec}$ with alphabet $2^{I\cup O}\times 2^\H$ that accepts all words $w\oplus w'_\H$ where $w\models\spec$ and $w'_\H$ is a certificate for $\sec$ being $\H$-hidden in $w$. Moreover, the UCW $\U^\H_{\spec,\sec}$ is of size $2^{O(|\spec|+|\sec|)}$. 
\end{lem}
\begin{proof}
	Let $\H'$ be a copy of the signals in $\H$, and let $\sec'$ be the LTL formula obtained from $\sec$ when replacing every signal in $\H$ with its copy in $\H'$. Consider the LTL formula $\Theta=\spec\wedge (\sec\leftrightarrow \neg \sec')$ above $I\cup O\cup \H'$. Since letters in $2^{I\cup O}\times 2^\H$ correspond to assignments to $I\cup O\cup \H'$, it follows immediately from the definition of $\Theta$ that  $\U^\H_{\spec,\sec}$ can be defined as a UCW for $\Theta$. The latter can be constructed by constructing an NBW for $\neg\Theta$ and then dualizing into the required UCW. By \cite{VW94}, such an NBW can be defined with size $2^{O(|\Theta|)}=2^{O(|\spec|+|\sec|)}$, and so we are done.
\end{proof}

\begin{thm}[LTL Synthesis with proved privacy]\label{2exp ltl certificate}
	Given two disjoint finite sets $I$ and $O$, LTL formulas $\spec$ and $\sec$ over $I\cup O$, and a set of hidden signals $\H\subseteq I\cup O$, deciding whether $\zug{\spec,\sec,\H}$ is realizable with certified privacy is 2EXPTIME-complete, and can be decided using a Safraless procedure. 
	\end{thm}
\begin{proof}
	The proof is very similar to the proof of \autoref{2exp ltl} with the only difference that the upper bound goes through UCW synthesis. By \autoref{dpw with certificates} there exists a UCW $\U^\H_{\spec,\sec}$ of exponential size that recognizes all words $w\oplus w'_\H\in (2^{I\cup O}\times 2^\H)^\omega$ that satisfy $\Theta=\spec\wedge (\sec\leftrightarrow\neg \sec')$ as defined in the proof of \autoref{dpw with certificates}. The first conjunct asserts that $w$ satisfy the specification $\spec$. Then, the conjunct $\sec\leftrightarrow\neg \sec'$ asserts that $w$ disagrees with the alternative computation $w|_\V\oplus w'_\H$ on the secret $\sec$, and thus $w'_\H$ is a certificate for $\sec$ being $\H$-hidden in $w$. We then consider the set $O'=O\cup \H'$ as a new set of output signals, and reduce the problem of synthesizing a system that realizes $\zug{\spec,\sec,\H}$ with certified privacy to the UCW synthesis of an $(I/O')$-transducer that realizes $\U^\H_{\spec,\sec}$ in the traditional sense. Note that as elaborated in the proof of \autoref{2exp ltl certificate}, assignments to $O'$ correspond to pairs of assignments, one to $O$ and another to $\H$, and so the alphabet $2^{I\cup O}\times 2^\H$ of $\U^\H_{\spec,\sec}$ is compatible with an $(I/O')$-transducer. In \cite{KV05c} the authors show that UCW synthesis can be solved in exponential time while avoiding Safra's determinization procedure (in fact it avoids determinization completely), and so we obtain the required upper bound. 
	
	Finally, the correspondence between $(I/O')$-transducers that realize $\U^\H_{\spec,\sec}$ and certifying transducers that realize $\zug{\spec,\sec,\H}$ with certified privacy is straight forward. Indeed, a labeling function $\tau':S\rightarrow 2^{O'}$ of an $(I/O')$-transducer, correspond to pair of labeling functions $\zug{\tau,c}$ of type $\tau:S\rightarrow 2^O$ and $c:S\rightarrow 2^\H$. Indeed, the natural bijection between $2^{O'}$ and $2^O\times 2^\H$ induces the needed decomposition of $\tau'$ to a pair $\zug{\tau,c}$. Formally, $\tau$ is simply the restriction of $\tau'$ to the signals in $O$ and $c$ is the restriction of $\tau'$ to the signals in $\H'$, where every signal $h\in \H$ is identified with its copy $h'\in \H'$. Thus, $c(h)=\tau(h')$ for all $h\in \H$. 
\end{proof}

\section{When the Observer Knows the Specification or the Transducer}
In this section we study two settings in which the observer has some knowledge that may help her to evaluate the secret. In the first setting, the observer knows the specification, and in the second, she knows the transducer. 

\subsection{An observer that knows the secret}
Since all the computations of the system satisfies the specification, an observer that knows the specification $\spec$ knows that only computations that satisfy $\spec$ should be taken into account when she tries to evaluate the secret. If, for example, $\spec\rightarrow\sec$, then $\sec$ cannot be kept private in a setting in which the observer knows $\spec$. Indeed, the fact $\spec$ is realized by the system reveals that $\sec$ is satisfied.
Formally, we say that $\T$ \emph{realizes $\zug{\spec,\sec, \H}$ with privacy under the knowledge of the specification}\/ if $\T$ realizes $\spec$, and for every $w_I\in (2^I)^\omega$, there exist $w^+,w^-\in \noiseH(\T(w_I))\cap L_\spec$ such that $w^+\models\sec$ and $w^-\not\models\sec$. Thus, the satisfaction of the secret $\sec$ in a computation $\T(w_I)$ cannot be deduced from the observable computation $\hideH(\T(w_I))$ even when the observer knows that $\spec$ is satisfied in $\T(w_I)$. The adjustment for the definition of the problem with respect to a hiding cost function and a budget is similar. 

We start by showing the analogue of Lemma~\ref{dpw} for the setting in which the observer knows the specification. 
The construction is similar to that of Lemma~\ref{dpw}, except that now, the construction of the DPW $\D^\H_{\sec|\spec}$ involves an existential projection on $\H$ also in the automaton for the specification. Accordingly, the size of the DPW is exponential in both the specification and the secret even in the case they are given by DBWs. 

\begin{lem}\label{dpw know spec}
	Consider two disjoint finite sets $I$ and $O$, a subset $\H\subseteq I\cup O$, and regular languages $L_\spec$ and $L_\sec$ over the alphabet $2^{I \cup O}$. There exists a DPW $\D^\H_{\sec|\spec}$ with alphabet $2^{I\cup O}$ that accepts a computation $w\in (2^{I\cup O})^\omega$ iff $w \in L_\spec$ and there exist $w^+,w^-\in \noiseH(w)$ such that $w^+ \in L_\spec\cap L_\sec$ and $w^-\in L_\spec\setminus L_\sec$. 
	\begin{enumerate}
		\item
		If $L_\spec$ and $L_\sec$ are given by \LTL formulas $\spec$ and $\sec$, then $\D^\H_{\sec|\spec}$ has $2^{2^{O(|\spec|+|\sec|)}}$ states and index $2^{O(|\spec|+|\sec|)}$. 
		\item
		If $L_\spec$ and $L_\sec$ are given by DPWs $\D_\spec$ and $\D_\sec$ with $n_\spec$ and $n_\sec$ states,
		and of indices $k_\spec$ and $k_\sec$,
		then $\D^\H_{\sec|\spec}$ has $2^{O((n_\spec\cdot k_\spec)^3\cdot (n_\sec\cdot k_\sec)^2 \log (n_\spec\cdot k_\spec\cdot n_\sec\cdot k_\sec) )}$ states and index $O((n_\spec\cdot k_\spec)^3\cdot (n_\sec\cdot k_\sec)^2)$.
	\end{enumerate}
\end{lem}

Lemma~\ref{dpw know spec} implies that all the asymptotic upper bounds described in Section~\ref{solve syn priv} are valid also in a setting with an observer that knows the specification. Also, as the lower bounds in Theorems~\ref{2exp ltl}  and~\ref{exp dbw} involve secrets that are independent of the specification, they are valid for this setting too. Two issues require a consideration: 
\begin{enumerate}
	\item
	The need to search for ${\mathcal H}$: the NP-hardness proof in Theorem~\ref{np} is no longer valid, as there, $\spec\rightarrow \lnot\sec$, and so the satisfaction value of the secret is revealed in a setting with an observer that knows the specification. 
	\item
	The construction in Lemma~\ref{dpw know spec} results in an algorithm that is exponential also in the specification, even when given by a DBW. On the other hand, the EXPTIME-hardness proof in \autoref{lb dbw hiding} does not imply an exponential lower bound in the specification. 
\end{enumerate}

Below we address the two issues, providing lower bounds for a setting in which the observer knows the specification. Matching upper bounds follow the same considerations in Theorems~\ref{exp dbw} and~\ref{np}, where $\D^\H_{\sec|\spec}$ replaces $\D^\H_{\spec,\sec}$.
We start with a variant of Theorem~\ref{np}, showing NP-hardness also in the setting of a knowledgeable observer. As mentioned above, the lower bound in the proof of \autoref{np} does not work when the observer knows the specification, yet,  can easily be modified to work for the case of a knowledgeable observer.
\begin{thm}\label{np knowledge}
	Given a set $O$ of output signals, a cost function $\cost:O\rightarrow \Nat$, a hiding budget $b\in \Nat$, and DBWs $\A_\spec$ and $\A_\sec$ over $2^O$, deciding whether there is $\H\subseteq O$, with $\cost(\H)\leq b$, such that $\zug{\A_\spec,\A_\sec,\H}$ is realizable with privacy under knowledge of the specification is NP-hard. Moreover, hardness holds already when $\cost$ is uniform.
\end{thm}

\begin{proof}
	We describe a polynomial-time reduction from the vertex cover problem.
	Consider an undirected graph $G=\zug{V,E}$, with $E=\set{e_0,e_2,\ldots,e_m}$. We construct DBWs $\A_\spec$ and $\A_\sec$ over the alphabet $2^V$, such that for all $\H\subseteq V$, it holds that $\zug{\A_\spec,\A_\sec,\H}$ is realizable with privacy under knowledge of the specification iff $\H$ is a vertex cover of $G$.
	Let $L^1_\spec=\set{\varnothing^\omega}$ and $L^2_\spec$ be the set of all words $w=\sigma_0\cdot \sigma_1\cdot \sigma_2\cdots\in (2^{V})^\omega$ such that $\sigma_i\cap e_i\neq \varnothing$, for all $0\leq i\leq m$. Now, let $L_\spec=L^1_\spec\cup L^2_\spec$, and $L_\sec=L^2_\spec$. It is not hard to see that there are DBWs $\A_\spec$ and $\A_\sec$ with $O(|E|)$ states that recognize $L_\spec$ and $L_\sec$, respectively, and can be constructed in polynomial time.
	We argue that $\zug{\A_\spec,\A_\sec,\H}$ is realizable with privacy under knowledge of the specification iff $\H$ is a vertex cover of $G$. Thus, $G$ has a vertex cover of size $b$ iff $\zug{L_\spec,L_\sec,\cost,b}$ is realizable with privacy under the knowledge of the specification, for the uniform cost function that assign $1$ to all $v\in V$.

	For the first direction, consider a vertex cover $\H$ of $G$. Let $\T$ be a transducer that generates the computation ${\varnothing}^{\omega}$ and hides the signals in $\H$.
	In every step of the computation, the observer sees that the signals in $V\setminus \H$ are assigned $\False$.
	Assigning $\False$ to the signals in $\H$ results in a word $w^-$ that is in $L^1_\spec\setminus L^2_\spec=L_\spec\setminus L_\sec$.
	On the other hand, as $\H$ is a vertex cover, and so it intersects every edge $e\in E$, then assigning $\True$ to the signals in $\H$ results in a word $w^+$ that is in $L^2_\spec = L_\spec\cap L_\sec$. 
	Hence, the membership of the computation in $L_\sec$ is not known to the observer, even when she knows it is in $L_\spec$. 
	That is, $\T$ raelizes $\zug{\A_\spec,\A_\sec,\H}$ with privacy under knowledge of the specification.
	
	In the other direction, consider a set $\H \subseteq V$ that is not a vertex cover of $G$.
	Consider a transducer $\T$ that realizes $L_\spec$ and hides the signals in $\H$. Let $w=\sigma_0\cdot \sigma_1\cdot \sigma_2\cdots\in (2^V)^\omega$ be the computation generated by the transducer. Since $\H$ is not a vertex cover of $G$, there is $0\leq i\leq m$ such that $\H\cap \{e_i\}=\varnothing$. Let $\kappa_i=\hideH(\sigma_i)$. We claim that the membership of $w$ to $L_\sec$ can be revealed from $\kappa_i$ when the observer knows that $w\in L_\spec$. Indeed, if $\kappa_i\cap e_i\neq \varnothing$, then $\sigma_i\neq \varnothing$ and so $w\neq \varnothing^\omega$. Hence, $w\in L_\spec\setminus L^1_\spec=L_\sec$. Also, if $\kappa_i\cap e_i=\varnothing$, then sa $\H\cap e_i=\varnothing$, we have that $\sigma_i\cap e_i=\varnothing$. Thus, $w\in L_\spec\setminus L^2_\spec=L_\spec\setminus L_\sec$.
\end{proof}

We continue to the second issue, proving that synthesis with privacy under knowledge of the specification is EXPTIME-hard even for specifications in DBWs and secrets of a fixed size. Note that synthesis with privacy (without knowledge of the specification) can be solved in PTIME in this case (see Remark~\ref{when dbw}). 
The proof is similar to that of \autoref{lb dbw hiding}, except that here the lower bound needs the secret to be a of a fixed size, making the specification more complex. It follows that the exponential blow-up in $\D_\spec$, which exists in Lemma~\ref{dpw know spec} cannot be avoided even when $\D_\spec$ is a deterministic B\"uchi (rather than parity) automaton and $\D_\sec$ is of a fixed size.

\begin{thm}\label{exp fixed} 
	Given two disjoint finite sets $I$ and $O$, DBWs $\D_\spec$ and $\D_\sec$ over $2^{I\cup O}$, and a set $\H\subseteq I\cup O$ of hidden signals, deciding whether $\zug{\D_\spec,\D_\sec,\H}$ is realizable with privacy under knowledge of the specification is EXPTIME-hard already when $\D_\sec$ is of fixed size.
\end{thm}

\begin{proof}
	We describe a polynomial-time reduction from NBW realizability: given an NBW $\A$ over $2^{I\cup O}$, we construct a new set of signals $\H$, and DBWs $\D_\spec$ and $\D_\sec$ over the alphabet $2^{I\cup O\cup \H}$, such that $\D_\sec$ has only two states and $L(\A)$ is realizable iff  $\zug{\D_\spec,\D_\sec,\H}$ is realizable with privacy under knowledge of the specification. 
	
	As in the proof of Theorem~\ref{lb dbw hiding}, the signals in $\H$ encode the states of $\A$, and thus we refer to words in $(2^{I\cup O\cup \H})^\omega$ as a combination $w\oplus r\in (2^{I\cup O}\times Q)^\omega$, where $w\in (2^{I\cup O})^\omega$ and $r\in Q^\omega$. Also, we assume that $\A$ has a single initial state and that every word in $(2^{I\cup O})^{\omega}$ has at least one rejecting run in $\A$. 
	Consider a word $w \oplus r \in (2^{I\cup O}\times Q)^\omega$. The key difference between the proof of Theorem~\ref{lb dbw hiding} and the one here, is that now the specification requires the $Q$-component $r$ to be a legal run of $\A$ on the  $2^{I\cup O}$-component $w$. Then the secret only requires $r$ to be accepting. Thus, $\D_\sec$ only needs two states to detect infinitely many visits in accepting states of $\A$.\footnote{Note that while $\D_\psi$ only has two states, its representation requires an encoding of $Q$. If one wants to get rid of this, it is possible to add a new signal $p_{\it acc}$ to $\H$, and use the specification in order to require $p_{\it acc}$ to hold exactly when the encoded state is accepting. That way, $\D_\sec$ has two states and its transitions depend only on $p_{\it acc}$.}	
	
	We prove that $L(\A)$ is realizable iff  $\zug{\D_\spec,\D_\sec,\H}$ is realizable with privacy under knowledge of the specification. Indeed, as in the proof of Theorem~\ref{lb dbw hiding}, a transducer $\T$ realizes $L(\A)$ iff $\T$ realizes $\zug{\D_\spec,\D_\sec,\H}$ with privacy under knowledge of the specification.
\end{proof}

\subsection{An observer that knows the transducer}

Recall that an observer that knows the specification can evaluate the secret only with respect to computations that satisfy the specification. In fact, the observer can do better, and restricts the search to computations that are generated by an $(I/O)$-transducer that realizes the specification. In this section we show that such a restriction is indeed stronger, and study an even stronger setting, in which the observer not only knows that computations are generated by an $(I/O)$-transducer that realizes the specification, but actually knows the $(I/O)$-transducer $\T$. Note that even though the observer knows all the components of $\T$, she still may not be able to evaluate the secret, as she sees only the signals not in $\H$. Note also that since $\T$ realizes the specification $\spec$, then $L(\T) \subseteq L_\spec$. As an observer that knows $\T$ also knows $L(\T)$, the latter implies that knowing $\T$ is more helpful for the observer than knowing the specification. 

We first demonstrate that an observer that takes into an account the fact that the computations are generated by an $(I/O)$-transducer that realizes $\spec$ can reveal secrets that are not revealed by an observer that only takes into account the fact that all the computations satisfy $\spec$. In order to see the difference between the two types of observers, consider the case $I=\H=\{p_1,p_2\}$, $O=\{q\}$, $\varphi=(q\leftrightarrow p_1) \lor  Gp_2$, and $\psi=p_1$.
An observer that only knows that all the computations satisfy $\spec$ does not know which of its two disjuncts is satisfied, and thus, even though she observes $q$, the value of $p_1$ stays secret. Formally, a transducer that realizes $\phi$ $\H$-hides $\psi$ from the observer, even if the observer knows that $\varphi$ is satisfied. Indeed,  for every observable computation $\kappa \in 2^{\{q\}}$, there is a computation $w^+\in \noiseH(\kappa)$ that satisfies $p_1\land Gp_2$ and a computation $w^-\in \noiseH(\kappa)$ that satisfies $(\neg p_1)\land Gp_2$. Hence, $\zug{\spec,\sec,\H}$ is realizable with privacy even when the observer knows the specification. 

On the other hand, an observer that knows that the computations are generated by a transducer $\T$ that realizes $\spec$ can do more. To see this, note that the specification $\varphi$ is {\em open equivalent} to the formula $q \leftrightarrow p_1$: for every transducer $\T$, we have that $\T$ realizes $\varphi$ iff $\T$ realizes $q \leftrightarrow p_1$. Indeed, if $\T$ does not realize $q \leftrightarrow p_1$, then $\varphi$ is not satisfied in computations that do not satisfy $Gp_2$, which is the case for almost all the computations of $\T$. Accordingly, an observer that knows that the computations are generated by a transducer that realizes $\spec$ also knows that they are generated by a  transducer that realizes $q \leftrightarrow p_1$, and so she can learn the value of the secret $p_1$ by observing the value of $q$.\footnote{Note that the example uses the fact that the specification $\varphi$ is {\em inherently vacuous}; namely that it is open equivalent to a simpler formula \cite{GBJV08,FKSV08}. The question whether there are also examples in which the two types of an observer differ for a specification that is not inherently vacuous is left open.}

We continue to an observer that knows $\T$. Formally, we say that $\T$ $\H$-hides $\sec$ from an observer that knows $\T$ if for all $w_I\in (2^I)^\omega$, there is $w'_I\in (2^I)^\omega$ such that $\hideH(\T(w_I))=\hideH(\T(w'_I))$ and the two computations $\T(w_I)$ and $\T(w'_I)$ do not agree on $\sec$.

Note that $\T$ does not $\H$-hide $\sec$ from an observer that knows $\T$ if there is $w_I\in (2^I)^\omega$ such that $\noiseH(\T(w_I))\cap L(\T)\subseteq\sec$ or $\noiseH(\T(w_I))\cap L(\T)\subseteq\neg\sec$.
Indeed, if such a word $w_I$ exists, then the set $\noiseH(\T(w_I))\cap L(\T)$, namely the set of all the computations that are possible interactions of $\T$ with $w_I$ from the point of view of the observer, reveals the satisfaction value of the secret $\sec$. Technically, rather than considering the set $\noiseH(\T(w_I))$ of computations, we consider its restriction $\noiseH(\T(w_I))\cap L(\T)$ to computations generated by $\T$. As discussed above, since $L(\T) \subseteq L_\spec$, such a restriction is stronger than the one induced by knowing $\spec$.  

We consider the problem of checking whether $\T$ $\H$-hides $\sec$ from an observer that knows $\T$.
The input to the problem is a transducer $\T$, an LTL formula $\sec$, and a set $\H\subseteq I\cup O$, and the mission is to decide whether $\T$ $\H$-hides $\sec$ from an observer that knows $\T$.

\begin{thm}
	Given an $(I/O)$-transducer $\T$, an LTL formula $\sec$, 
	and a set $\H\subseteq I\cup O$ of hidden signals, checking whether $\T$ $\H$-hides $\sec$ from an observer that knows $\T$ can be solved in EXPSPACE. 
\end{thm}

\begin{proof}
	Consider an LTL formula $\sec$, a set $\H\subseteq I\cup O$, and an $(I/O)$-transducer $\T=\zug{I,O,S,s_0,\mu,\tau}$.
	Let $\U_\sec=\zug{2^{I\cup O}, Q, q_0, \delta, \alpha}$ be a UCW for $\sec$.
	There is such a UCW of size exponential in $|\sec|$ \cite{VW94}.
	We construct two UCWs $\U^+$ and $\U^-$ over the alphabet $2^I$, 
	such that $L(\U^+)=\{w_I : \noiseH(\T(w_I))\cap L(\T)\subseteq\sec\}$ and $L(\U^-)=\{w_I : \noiseH(\T(w_I))\cap L(\T)\subseteq\neg\sec\}$.
	Then, $\T$ $\H$-hides $\sec$ from an observer that knows $\T$ iff $L(\U^+)\neq\varnothing$ or $L(\U^-)\neq\varnothing$.
	
	We now describe the UCW $\U^+$.
	Reading a word $w_I\in 2^I$, the UCW $\U^+$ traces every word $w'_I$ with $\T(w'_I)\in\noiseH(\T(w_I))\cap L(\T)$.
	The run corresponding to $w'_I$ is accepting iff $\T(w'_I)\in\sec$.
	Thus, all of the runs on $w_I$ are accepting iff $\noiseH(\T(w_I))\cap L(\T)\subseteq\sec$.
	The state space of $\U^+$ is $S\times S\times Q$.
	The first $S$-component of each state is used to simulate the run of $\T$ on $w_I$.
	The second component$S$-component is used to simulate a run of $\T$ on a word in $\noiseH(\T(w_I))$.
	That is, after reading a word $x\in (2^I)^*$, each run of $\U^+$ reaches a state of the form $\zug{\mu^*(s_0,x),s_2,q}$, where $s_2= \mu^*(s_0,x')$ for some word $x'\in 2^I$ with $\T(x')\in\noiseH(\T(x))$.
	The $Q$-component of each state of $\U^+$ simulates the runs of $\U_\sec$ on words in $\noiseH(\T(w_I))$.
	That is, $q\in\delta^*(q_0, \T(x'))$ for some word $x'\in 2^I$ with $\T(x')\in\noiseH(\T(x))$.
	
	Formally, $\U^+=\zug{2^I,S\times S\times Q, \zug{s_0, q_0},\delta', \alpha'}$, 
	where for all $\zug{s_1,s_2,q}\in S\times S\times Q$ and $i\in 2^I$, we have that $\zug{s_1,s_2,q}\in\alpha'$ iff $q\in\alpha$, and $\delta'(\zug{s_1,s_2,q},i)=\{\zug{s_1',s_2',q'} : s_1'=\mu(s_1,i)\wedge \exists i'\in\noiseH(i)\text{ s.t. }	(s_2'=\mu(s_2,i')\wedge\tau(s_2')\in\noiseH(\tau(s_1'))\wedge q'\in\delta(q,i'\cup\tau(s_2')))\}$.
	
	We prove that the language of $\U^+$ is exactly all words $w_I$ with $\noiseH(\T(w_I))\cap L(\T)\subseteq\sec$.
	For the first direction, consider a word $w_I=i_0\cdot i_1\cdot\ldots\in L(\U^+)$.
	That is, all runs of $\U^+$ on $w_I$ are accepting.
	Let $w'=(i'_0\cup o'_0)\cdot (i'_1\cup o'_1)\cdot\ldots\in \noiseH(\T(w_I))\cap L(\T)$.
	We show that $w'\in\sec$.
	Let $q_0\cdot q_1\cdot\ldots$ be a run of $\U_\sec$ on $w'$.
	We show that this run is accepting.
	Let $s_0\cdot s_1\cdot\ldots$ be the run of $\T$ on $w_I$, and let $s'_0\cdot s'_1\cdot\ldots$ be the run of $\T$ on $w'|_I=i'_0\cdot i'_1\cdot\ldots$.
	Since $w'\in L(\T)$, we have that $o'_j=\tau(s'_{j+1})$, for all $j\geq 0$.
	Also, since $w'\in \noiseH(\T(w_I))$, we have that $i'_j\in\noiseH(i_j)$ and $o'_j\in\noiseH(\tau(s_{j+1}))$, for all $j\geq 0$.
	Thus, by definition of $\U^+$, we have that $\zug{s_0, s'_0,q_0}\cdot \zug{s_1, s'_1,q_1}\cdot\ldots$ is a run of $\U^+$ on $w_I$.
	Since all runs of $\U^+$ on $w_I$ are accepting, we get that $q_j\in\alpha$ for finitely many $j$-s.
	Thus, the run $q_0\cdot q_1\cdot\ldots$ is accepting as needed.
	
	For the second direction, consider a word $w_I=i_0\cdot i_1\cdot\ldots\in(2^I)^\omega$ with $\noiseH(\T(w_I))\cap L(\T)\subseteq\sec$.
	We show that $w_I\in L(\U^+)$.
	Let $r=\zug{s_0, s'_0,q_0}\cdot \zug{s_1, s'_1,q_1}$ be a run of $\U^+$ on $w_I$. 
	We show that this run is accepting.
	First, by the definition of $\U^+$, we have that $s_0\cdot s_1\cdot\ldots$ is the run of $\T$ on $w_I$.
	For all $j\geq 0$, let $o_j$ be $\tau(s_{j_1})$.
	By the definition of $\U^+$, there is a word $w'=(i'_0\cup o'_0)\cdot (i'_1\cup o'_1)\cdot\ldots$ with $i'_j\in\noiseH(i_j)$, $o'_j=\tau(s'_{j+1})$ and $o'_j\in\noiseH(o_j)$ for all $j\geq 0$, for which $q_{j+1}\in\delta(q_j, i'_j\cup o'_j)$.
	Since $o'_j=\tau(s'_{j+1})$ for all $j\geq 0$, we get that $w'\in L(\T)$.
	Since $i'_j\in\noiseH(i_j)$ and $o'_j\in\noiseH(o_j)$ for all $j\geq 0$, we get that $w'\in \noiseH(\T(w_I))$.
	Overall, $w'\in\noiseH(\T(w_I))\cap L(\T)$, and by that $w'\in\sec$.
	Moreover, since $q_{j+1}\in\delta(q_j, i'_j\cup o'_j)$ for all $j\geq 0$, we get that $q_o\cdot q_1\cdot\ldots$ is a run of $\U_\sec$ on $w'$, and thus accepting.
	That is, $q_j\in\alpha$ for only finitely many $j$-s, and so $r$ is accepting as needed.

	We construct the UCW $\U^-$ in a similar way, using $\U_{\neg\sec}$ instead of $\U_\sec$. Recall that the observer wins iff $L(\U^+)\neq\varnothing$ or $L(\U^-)\neq\varnothing$, which can be decided by checking nonemptiness twice, once for $\U^{+}$ and once for $\U^-$. Since universality of UCWs can be solved in PSPACE \cite{SVW87}, and the UCWs $\U^{+}$ and $\U^-$ are of size exponential in the size of $\sec$, the induced algorithm is in EXPSPACE.
\end{proof}

Readers familiar with HyperLTL \cite{CFKMRS14, FHLST20}, might notice that $\H$-hiding from an observer that knows $\T$ has the flavor of a hyper property.
Indeed, a different approach to decide whether $\T$ $\H$-hides a secret from an observer that knows $\T$, is by model checking of hyper properties:
Given an LTL formula $\sec$, consider the hyper LTL formula $f=\forall\pi\exists\pi' (\bigwedge_{v\in V} v_\pi\leftrightarrow v_{\pi'})\wedge(\sec(\pi)\leftrightarrow\neg\sec(\pi'))$,
where $\sec(\pi)$ and $\sec(\pi')$ are  the LTL formulas obtained from $\sec$ by replacing each signal $p$ by the path variables $p_\pi$ and $p_{\pi'}$, respectively.
It is easy to see that $\T$ $\H$-hides $\sec$ from an observer that knows $\T$, if and only if $\T\models f$.
Indeed, $\T\models f$ if and only if for every computation $\pi\in L(\T)$, there is a computation $\pi'\in L(\T)$ that agrees with $\pi$ on the visible signals, but does not agree with $\pi$ on the secret $\sec$.
That is, our problem can be reduced to model checking of HyperLTL formula, which is decidable \cite{FHL20}.

The formulation of $\H$-hiding from an observer that knows $\T$ as a HyperLTL formula of the form $\forall \pi \exists \pi' f(\pi,\pi')$ suggests that synthesis with privacy against an observer that knowns the transducer is undecidable, as is synthesis of HyperLTL formulas of this form \cite{CFKMRS14}.

\section{Directions for Future Research}

We suggested a framework for the synthesis of systems that satisfy their specifications while keeping some behaviors secret. Behaviors are kept secret from an observer by hiding the truth value of some input and output signals, subject to budget restrictions: each signal has a hiding cost, and there is a bound on the total hiding cost.
Our framework captures settings in which the choice and cost of hiding are fixed throughout the computation. For example, settings with signals that cannot be hidden (e.g., alarm sound, or the temperature outside), signals that can be hidden throughout the computation with some effort (e.g., hand movement of a robot), or signals that are anyway hidden (e.g., values of internal control variables). Our main technical contribution are lower bounds for the complexity of the different aspects of privacy: the need to choose the hidden signals, and the need to hide the secret behaviors. We show that both aspects involve an exponential blow up in the complexity of synthesis without privacy. 

The exponential lower bounds apply already in the relatively simple cost mechanism we study. Below we discuss  possible extensions of this mechanism. In settings with a {\em dynamic hiding of signals}, we do not fix a set ${\mathcal H} \subseteq I \cup O$ of hidden signals. Instead, the output of the synthesis algorithm contains a transducer that describes not only the assignments to the output signals but also the choice of input and output signals that are hidden in the next cycle of the interaction. Thus, signals may be hidden only in segments of the interaction -- segments that depend on the history of the interaction so far. For example, we may hide information about a string that is being typed only after a request for a password. In addition, the cost function need not be fixed and may depend on the history of the interaction too. For example, hiding the location of a robot may be cheap in certain sections of the warehouse and expensive in others. Solving synthesis with privacy in a setting with such dynamic hiding and pricing of signals involves automata over the alphabet $3^{I \cup O}$, reflecting the ability of signals to get an ``unknown" truth value in parts of the computation. Moreover, as the cost is not known in advance (even when the cost of hiding signals is fixed), several mechanisms for bounding the budget are possible (energy, mean-payoff, etc. \cite{BBFR13,CD10}).

\bibliographystyle{alphaurl}
\bibliography{ok}

\end{document}